\documentclass[a4paper,twocolumn,superscriptaddress,11pt,longbibliography]{quantumarticle}

\usepackage[utf8]{inputenc}
\usepackage[english]{babel}
\usepackage[T1]{fontenc}

\usepackage{amsmath}
\usepackage{amsfonts}
\usepackage{amssymb}
\usepackage{amsthm}
\usepackage{hyperref}
\usepackage[numbers,sort&compress]{natbib}
\usepackage{tikz}
\usepackage{lipsum}
\usepackage[normalem]{ulem}
\usepackage{exscale}
\usepackage{bbm}
\usepackage{graphicx}
\usepackage{latexsym}
\usepackage{enumerate}
\usepackage{bbold}
\usepackage{color}
\usepackage{nicefrac}
\usepackage{mathtools}
\usepackage{chngcntr}
\usepackage{tcolorbox}
\usepackage{wrapfig}
\usetikzlibrary{decorations.pathreplacing}
\usetikzlibrary{graphs}
\usepackage{tikz}
\usepackage{pgfplots}
\usepackage{array, threeparttable, booktabs}

\usepackage{subcaption}
\theoremstyle{definition}
\newtheorem{definition}{Definition}
\theoremstyle{remark}

\definecolor{quantumviolet}{HTML}{53257F} 
\definecolor{quantumgray}{HTML}{555555} 
\definecolor{mygray}{gray}{0.95} 

\newtcolorbox[auto counter,number within=section]{boxfigure}[2][]{%
colback=mygray,colframe=quantumviolet,fonttitle=\bfseries,width=\textwidth,float*=ht,lower separated=false, halign=justify,title=Box~\thetcbcounter: #2,#1}

\usepackage{listings} 
\definecolor{commentsColor}{rgb}{0.497495, 0.497587, 0.497464}
\definecolor{keywordsColor}{rgb}{0.000000, 0.000000, 0.635294}
\definecolor{stringColor}{rgb}{0.558215, 0.000000, 0.135316}

\usepackage{tikz}
\usetikzlibrary{quantikz}
\usepackage{tikz-cd}
\tikzcdset{scale cd/.style={every label/.append style={scale=#1},
		cells={nodes={scale=#1}}}}
\usepackage{physics}
\usepackage[outline]{contour} 
\usetikzlibrary{intersections}
\usetikzlibrary{decorations.markings}
\usetikzlibrary{angles,quotes} 
\usetikzlibrary{calc}
\usetikzlibrary{3d}

\tikzset{>=latex} 
\colorlet{myred}{red!85!black}
\colorlet{myblue}{blue!80!black}
\colorlet{mycyan}{cyan!80!black}
\colorlet{mygreen}{green!70!black}
\colorlet{myorange}{orange!90!black!80}
\colorlet{mypurple}{red!50!blue!90!black!80}
\colorlet{mydarkred}{myred!80!black}
\colorlet{mydarkblue}{myblue!80!black}
\tikzstyle{xline}=[myblue,thick]

\tikzstyle{myarr}=[myblue!50,-{Latex[length=3,width=2]}]

\usepackage{lipsum}
\usepackage{adjustbox}
\usepackage{longtable}
\usepackage{colortbl}
\usetikzlibrary{arrows}
\usepackage{makeidx}\makeindex
\tikzset{
	operator/.append style={fill=purple!20},
	my label/.append style={above right,xshift=0.3cm},
	phase label/.append style={label position=above}
}

\usepackage{silence}
\definecolor{arsenic}{rgb}{0.23, 0.27, 0.29}
\definecolor{fluxcolor}{RGB}{204, 217, 255}
\definecolor{uwavecolor}{RGB}{244, 220, 222}
\definecolor{FandUwavecolor}{RGB}{231,244,224}
\definecolor{cavitycolor}{RGB}{232, 200, 244}
\definecolor{orangeb}{rgb}{0.99,0.78,0.07}
\definecolor{orangebdark}{rgb}{0.99,0.78,0.17}
\definecolor{livingcoral}{HTML}{FA7268}			
\definecolor{ultraviolet}{HTML}{5F4B8B}			
\definecolor{greenery}{HTML}{88B04B}			
\definecolor{radiantorchid}{HTML}{AD5E99}		
\definecolor{tangerinetango}{HTML}{DD4124}		

\DefineNamedColor{named}{Salmon}        {cmyk}{0,0.53,0.38,0}
\DefineNamedColor{named}{CarnationPink} {cmyk}{0,0.63,0,0}

\usepackage{float}

\usepackage{imakeidx}
\makeindex[columns=2, title=Alphabetical Index,  options= -s alpha.ist]

\definecolor{commentsColor}{rgb}{0.497495, 0.497587, 0.497464}
\definecolor{keywordsColor}{rgb}{0.000000, 0.000000, 0.735294}
\definecolor{stringColor}{rgb}{0.558215, 0.000000, 0.135316}
\definecolor{carrotorange}{rgb}{0.93, 0.57, 0.13}
\usepackage{tabularx}
\usepackage{ltablex}
\usepackage{caption}

\DeclareCaptionLabelFormat{boxed}{%
	\kern0.01em{\color{carrotorange}\rule{0.73em}{0.73em}}%
	\hspace*{0.47em}\bothIfFirst{#1}{~}#2}
\tikzstyle{input}=[draw,fill=red!50]
\tikzstyle{conv}=[draw,fill=black!20]
\tikzstyle{max}=[draw,dashed,fill=black!10]
\tikzstyle{dropout}=[draw,dashed,fill=blue!10]
\tikzstyle{fc}=[draw,fill=green!10]
\tikzstyle{output}=[draw,fill=red!50]

\usepackage[europeancurrents, europeanvoltages, americanresistors, cuteinductors, americanports, nosiunitx, noarrowmos]{circuitikz}

\definecolor{BlueLUH}{cmyk}{1.0,0.7,0,0}
\colorlet{LightBlue}{BlueLUH!20!white}
\colorlet{DarkBlue}{BlueLUH!80!black!20}
\colorlet{PinKish}{red!60}

\usepackage{tikz}
\usetikzlibrary{arrows,automata}

\usepackage{enumitem}
\setlist{nosep,leftmargin=\leftmargin/2}

\usepackage[utf8]{inputenc}
\usepackage[T1]{fontenc}
\usepackage[normalem]{ulem}

\usepackage{csquotes}

\usepackage{caption}
\usepackage{blochsphere}
\usepackage{braket}

\usepackage{algpseudocode,caption}
\DeclareCaptionType{algorithm}

\usepackage{changepage}

\usepackage{amsthm}
\usepackage{amsmath}
\usepackage{amsthm}
\theoremstyle{definition}
\newtheorem{Theorem}{Theorem}
\newtheorem{example}{Example}[section]

\DeclareRobustCommand{\numbercircle}[1]{%
	\tikz[baseline=(char.base)]{
		\node[shape=circle,draw,inner sep=2pt,fill=gray!20] (char) {\small #1};}%
}
\begin{document}

\index{Entanglement}

\title{$\langle$QonFusion$\rangle$- Quantum Approaches to Gaussian Random Variables: Applications in Stable Diffusion and Brownian Motion.}

\author{\raisebox{-0.2em}{\textcolor{DarkBlue}{\rule{0.8em}{0.8em}}} Shlomo Kashani \textsuperscript{$\Phi$}}
\affiliation{Applied Physics, Johns Hopkins University, Maryland U.S.A.: \url{skashan2@jhu.edu} }


\maketitle
\onecolumn
\begin{abstract}
In the present study, we delineate a strategy focused on non-parametric quantum circuits for the generation of Gaussian random variables (GRVs). This quantum-centric approach serves as a substitute for conventional pseudorandom number generators (PRNGs), such as the \textbf{torch.rand} function in PyTorch. The principal theme of our research is the incorporation of Quantum Random Number Generators (QRNGs) into classical models of diffusion. Notably, our Quantum Gaussian Random Variable Generator fulfills dual roles, facilitating simulations in both Stable Diffusion (SD) and Brownian Motion (BM). This diverges markedly from prevailing methods that utilize parametric quantum circuits (PQCs), often in conjunction with variational quantum eigensolvers (VQEs). Although conventional techniques can accurately approximate ground states in complex systems or model elaborate probability distributions, they require a computationally demanding optimization process to tune parameters. Our non-parametric strategy obviates this necessity. To facilitate assimilating our methodology into existing computational frameworks, we put forward QonFusion, a Python library congruent with both PyTorch and PennyLane, functioning as a bridge between classical and quantum computational paradigms. We validate QonFusion through extensive statistical testing, including tests which confirm the statistical equivalence of the Gaussian samples from our quantum approach to classical counterparts within defined significance limits. QonFusion is available at \url{https://boltzmannentropy.github.io/qonfusion.github.io/} to reproduce all findings here.

\end{abstract}

\setcounter{tocdepth}{4}
\tableofcontents


\onecolumn
\section{Introduction}

\subsection{The Convergence of Quantum Computing and Generative AI models}

Our investigations into QRNGs reveal a broad spectrum of potential applications. Chief among these is the ability for QRNGs to serve as direct replacements for traditional random number generators like \texttt{torch.rand} or \texttt{numpy.rand}, which are routinely employed in the realm of deep learning. Furthermore, QRNGs are effective in introducing Gaussian noise during the forward pass of SD models. They also offer valuable perspectives for simulating reversible Stochastic Differential Equations (SDEs), such as BM. A noteworthy discovery, referenced in section (\ref{sec:brownian}), is the finding that even a minimal deviation from a zero mean in the Gaussian Probability Density Function (PDF) during a BM simulation via QRNG results in irreversibility. This utility is especially relevant in the current scientific landscape where deep learning intersects extensively with quantum many-body systems, particularly in high-energy physics \cite{song2020score}. Simultaneously, the interface between quantum computing and quantum neural networks is attracting substantial academic interest.

While these research directions might initially appear disparate, they inherently exhibit complementary attributes, particularly in the realm of generative SD models \cite{gili2022evaluating}. Foundational contributions from Dhariwal et al. \cite{dhariwal2021diffusion} and Moghadam et al. \cite{moghadam2022morphology} have highlighted the efficacy of diffusion models in the creation of both general-purpose and medical images. Such models present a viable substitute for the labor-intensive and economically taxing process of manual tumor annotation.

Given these insights, an imperative question arises: Is it feasible to substitute all elements of a classical Machine Learning (ML) diffusion model with their quantum equivalents? The answer, although layered, is primarily negative. As a case in point, consider the U-Net architecture \cite{Ronneberger2015}, renowned for its original purpose in image segmentation. It has since been adapted to perform denoising functions within diffusion models. This dual-functionality architecture, comprising both an encoder and a decoder, can generate noise-highlighting masks that can be integrated into specific denoising algorithms. Yet, when one peels back the layers of the quantum realm, it quickly becomes apparent that the field is less developed than its classical counterpart.

To bridge our initial foray into QRNGs with the more nuanced world of Quantum Machine Learning (QML), it's crucial to establish the state-of-play in both the classical and quantum settings. QML, unlike its classical ML analogue, lacks well-defined quantum versions of key architectural elements such as encoders and decoders. This shortfall presents a significant barrier to the complete transition from classical to quantum frameworks in diffusion models. Consequently, our current research aims are situated in the early stages of classical diffusion models, specifically focusing on incorporating and simulating basic quantum circuits for QRNGs, as illustrated in Figure 1 \cite{franzese2022how,cao2023exploring}. Though this focus might appear limiting, it serves as a gateway to a broader, yet largely unexplored, landscape of opportunities.

Most of our experiments are conducted on classical computing platforms, harnessing the power of PennyLane's quantum simulator backends to serve as the computational engines for our quantum circuits. Optimised for computational efficiency, these backends are well-suited for numerical simulations in quantum computing, particularly those pertinent to SD or BM. Finally, it's important to note that our chosen default parameters are not intended to represent optimal configurations but rather to serve as a reference framework for evaluating Gaussian Random Variables (GRVs). A more detailed account of our research methodology will follow in subsequent sections of this paper.

In light of these developments, it becomes pertinent to explore how quantum computing can further enrich the field of generative models, particularly diffusion models. The following sections delve into this subject, examining the foundational principles of diffusion models and their potential for quantum augmentation.

\begin{figure}[h!]
	\centering
	\includegraphics[width=17.5cm]{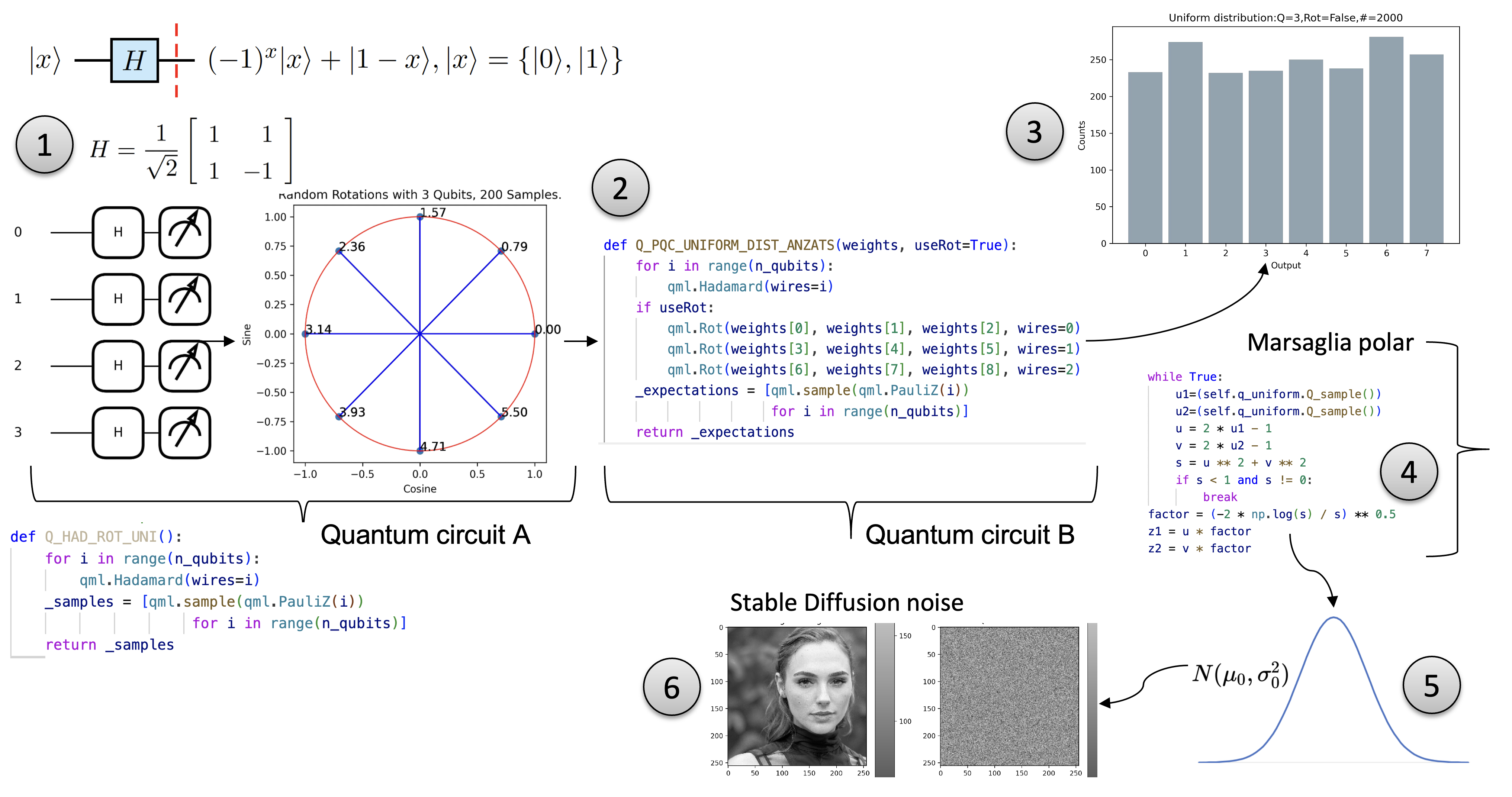}
	\caption[An illustrative Depiction of our Proposed Quantum Pipeline.]{An illustrative depiction of our proposed quantum pipeline. \numbercircle{1} Utilizing $N$ qubits, an initial non-parametric (\ref{sec:non-parametric}) quantum circuit, founded on Hadamard transformations, engenders a uniform distribution (\ref{sec:disc-uniform}) across $2^N$ discrete values through sampled measurement outcomes, rather than expectation values (further elucidated in \ref{fig:static_pqc_ansatz_sampling}). These outcomes are subsequently transmuted into equidistant angles, which are then channeled into a static quantum circuit that neither adapts nor learns, with the purpose of inducing stochastic rotations. \numbercircle{2} This static assembly, encompassing $M>N$ qubits, functions as follows: a Hadamard gate is applied to each qubit, thus creating a superposition of states. Upon activation of the \texttt{useRot} flag, the circuit administers a rotation (\ref{sec:rots}) to every qubit, utilizing the \textbf{previously generated angles as the parameters for these rotations}. The output \numbercircle{3} from this circuit is then employed in the \numbercircle{4} Marsaglia polar method to transmute the uniform distribution into a zero-mean Gaussian \numbercircle{5}. This transformation is achieved through the equation $Z = \sqrt{-2\log U}\cos(2\pi V)$, where $U$ and $V$ are two uniform pseudo-random numbers (PRN's), and by substituting our quantumly-generated uniform random bits for $U$ and $V$, we synthesize a source of quantum (\ref{sec:guss}) Gaussian random variables.  \numbercircle{6} The Marsaglia polar method serves as a paradigmatic example of how quantum random bits can seamlessly supplant classical PRN's.}
	\label{fig::our_pipeline}
\end{figure}

\subsection{A general overview of our pipeline}
In the hybrid quantum-classical framework we describe, as illustrated in Figure~\ref{fig::our_pipeline}, we make use of \textbf{two} individual quantum circuits in a sequential manner. The first circuit, which utilises \( N \) qubits, aims to produce a uniform distribution over \( 2^N \) unique states. These states are subsequently converted into equiangular values. These equiangular values are then fed into the second, pre-determined quantum circuit. This latter circuit, containing \( M > N \) qubits, applies a Hadamard gate to each qubit and carries out rotations as determined by the angles generated in the first circuit. The output from this series of quantum operations is then used in the classical Marsaglia polar method to transform the uniform distribution into a Gaussian distribution centred at zero. This serves as an illustrative example of how QRNGs can replace classical PRNGs in models like SD or BM, particularly during the forward diffusion stage. The numerical trials are conducted on idealised quantum simulators running on classical hardware, thereby approximating the essential elements of diffusion methodologies.

As we transition from the quantum-classical hybrid framework to a more focused discussion on diffusion models, it is crucial to remember that our ultimate aim is to explore the quantum-classical interface. The subsequent section on diffusion models should be viewed as an extension of this aim, offering a detailed background against which our quantum contributions can be more fully understood.

\subsection{Diffusion Models}
Our work is fundamentally anchored in the generation of pseudo-random sequences, a critical aspect of diffusion techniques that often takes up a significant portion of the computational resources. While this is a focal point of our study, it's worth noting that the applications of SD models are manifold, including, but not limited to, the realms of speech, image, and audio generation, as well as denoising techniques \cite{AbuHussein2022, Fishman2023, Yuan2022, Bai2023, Baas2022}. Given the extensive scope of these applications, a complete discussion in this article would be impractical.

\begin{figure}[h!]
	\centering
	\includegraphics[width=12.5cm]{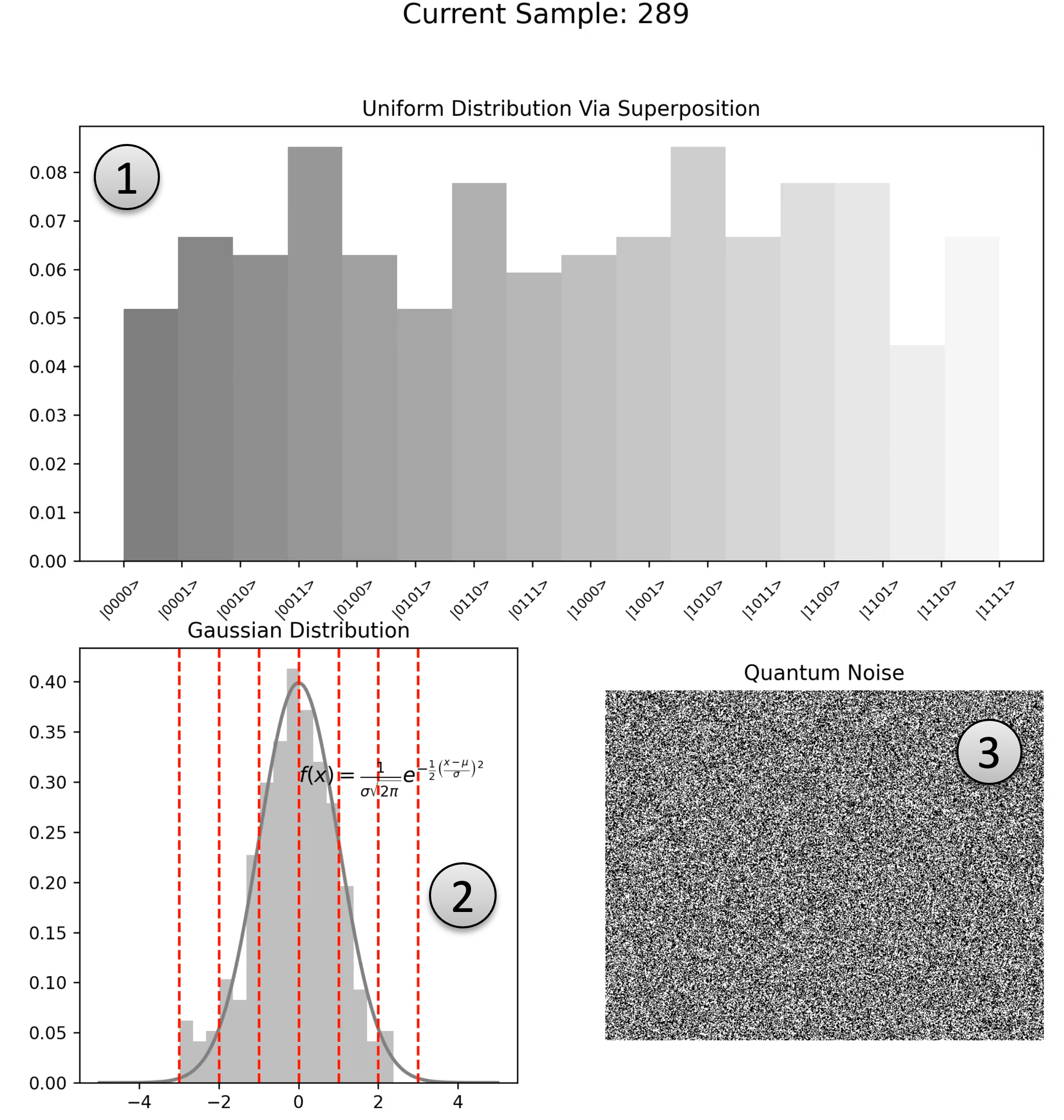}
	\caption[Quantum-Generated Gaussian Noise in Stable Diffusion Processes]{An illustrative schema elucidating the employment of quantum-sourced uniform random variables for the generation of Gaussian noise, specifically for the perturbation of images in a Stable Diffusion context. \numbercircle{1} Presents a histogram of the quantum-derived uniform random variables. \numbercircle{2} Demonstrates the Gaussian distribution fitted to the uniform variables, facilitated by the Marsaglia polar method. \numbercircle{3} Exhibits an image in two dimensions that has been disrupted by quantum-originated Gaussian noise.}
	\label{fig::sd_quant}
\end{figure}

Figure \ref{fig::sd_quant} captures a singular moment from an animated sequence that visualizes the operational stages of the QRNG. The complete animation can be viewed at \url{https://boltzmannentropy.github.io/qonfusion.github.io/}. This visualization serves as a practical illustration of \textbf{our central research theme: the integration of QRNGs into classical diffusion models}. While the field of diffusion models is vast and rich in applications, our focus remains on identifying areas where quantum computing, particularly QRNGs, can make a meaningful impact.

\subsubsection{Stable Diffusion}
Building on the theme of QRNG integration, the domain of diffusion models is rich in seminal contributions. Works by Sohl-Dickstein et al. \cite{Sohl-Dickstein2015} and Song et al. \cite{Song2019, Song2020, Song2021} stand as milestones in establishing the theoretical underpinnings of this field. Their groundbreaking research has set the stage for the work by Ho et al. \cite{Ho2020}, which offers a unified, practical framework that has become the go-to standard for multiple implementations. For readers seeking a deeper theoretical dive, the studies by Song et al. \cite{Song2020, Song2021} are invaluable. On the mathematical front, Luo et al. \cite{Luo2022} have compiled essential formulations for anyone interested in a comprehensive understanding of diffusion models. Furthermore, for an overview of the diversity of applications, one cannot overlook the work by Yang et al. \cite{Yang2022}.

\subsubsection{Brownian Motion}
\label{sec:brownian}
Having established the foundational principles and key contributions in the domain of SD models, it is instructive to delve into a specific, yet seminal, stochastic process that often serves as the underpinning of a different group of diffusion models. This leads us to the discussion of Brownian Motion, a process that not only has historical significance but also offers a robust mathematical framework for understanding diffusion dynamics.  BM is a cornerstone of stochastic processes and is unique for being the first to incorporate continuous time and state variables, thus leaving an indelible mark on subsequent research, notably in Gaussian processes, martingales, and Markov processes. Initially brought to light by Einstein \cite{Einstein1905}, BM provides a robust framework for the analytical representation of the random motion of particles in fluids, a subject further explored by Perrin \cite{Perrin1908}, Chandrasekhar \cite{Chandrasekhar1943}, and Langevin \cite{Langevin1908}. The process is governed by stochastic differential equations and, distinctively, employs Gaussian-distributed stochastic variables for its diffusion dynamics. This stands in contrast to other stochastic processes, which often resort to Poisson or Bernoulli distributions \cite{Gardiner2004, vanKampen2007}.

The reason we are testing our QRNG in BM is that BM is highly sensitive to any deviations in the mean of the Gaussian distribution, and this serves as a benchmark for our quantum generator. Our findings concerning the application of our QRNG to Brownian Motion simulations are detailed in Section \ref{res:brownian}, entitled \textit{Quantum-Augmented Gaussian Fluctuations in Simulations of Brownian Motion}.

%

\section{The Fundamentals of Classical Pseudo-Random Number Generation}
\label{sec:classical}
As a prelude to our exploration of quantum-generated random numbers, it is instructive to delineate the foundational methods employed in classical computing to produce uniform and Gaussian distributions \cite{AhrensDieter1972}. Classical Pseudo-Random Numbers (PRNs) are generated using deterministic algorithms that offer repeatability while producing a sequence of numbers that seemingly fall within the random range of \( (0,1) \) \cite{brent2010longperiod,brent1974gaussian}. The effectiveness of such generators is gauged through two primary criteria:

\begin{enumerate}
	\item \textbf{The Length of the Period}: This measure signifies the duration before the generated sequence repeats. Period lengths below \( 2^{32} \) are generally considered insufficient, and a minimum period length of \( 2^{40} \) is usually recommended for robust applications. For further insights into the nuances of such sequences, refer to the works by \cite{box1958note,brent1974gaussian}.
	\item \textbf{Robustness in Statistical Terms}: Any generator worth its salt will have undergone meticulous statistical evaluation to ascertain its quality of 'randomness'. Software that passes these stringent tests is considered reliable.
\end{enumerate}

Moving forward, we delve into the techniques used to generate numbers following a non-uniform, specifically Gaussian, distribution. The mathematical expression for the probability density function (PDF) of such a Gaussian distribution can be articulated as:

\begin{equation}
	P(x) = \exp\left(-\frac{(x - x_0)^2}{2 \sigma^2}\right)
\end{equation}

The Central Limit Theorem comes to the fore here. The theorem posits that the summation of a large number of independent random variables tends to a Gaussian distribution, with its width inversely related to \(\sqrt{N}\). Therefore, it's not a stretch to infer that by summing \( n \) uniform random numbers, one can approximate a Gaussian distribution. The closer \( n \) is to a large number, the more faithfully the resulting distribution will mirror a Gaussian form. To customise this distribution with a particular mean \(\bar{x}\) and width \(\sigma\), the following transformation of the summation \( S \) of \( n \) uniform random variables can be employed:

\begin{equation}
	x = \bar{x} - 2 \sigma \left(\frac{S}{n} - \frac{1}{2}\right) \sqrt{3}
\end{equation}

It's worth mentioning that the traditional approach of approximating a Gaussian distribution through the summation of \( n \) uniform random variables is not computationally efficient. This is because generating a single Gaussian Random Variable (GRV) \cite{brent2010longperiod,billard1985computer,brent1974gaussian} requires the prior generation of \( n \) uniform random variables within the range \( (0,1) \). While classical methods for generating these uniform random variables abound, our attention is rather riveted on how to transform quantum-generated uniform random numbers into GRVs via classical computational techniques. Here, we elucidate four key methods for this purpose:

\begin{enumerate}
	\item \textbf{The Box-Muller Technique} \label{eq:box_muller}: This method stands out for its computational efficiency and is well-adopted in practice. It takes two independent uniform random variables as inputs and produces two independent standard GRVs as outputs \cite{Michelen2018}.

	\item \textbf{Marsaglia's Polar Approach} \label{eq:marsaglia_polar}: This is essentially a variant of the Box-Muller method but avoids the use of trigonometric functions, offering a quicker alternative in certain scenarios \cite{Hwang2013}.

	\item \textbf{Marsaglia's Ziggurat Method} \label{eq:marsaglia_ziggurat}: Known for its efficiency in generating normal and exponential random variables, this method employs precomputed tables, making it especially useful for large-scale simulations \cite{Nadler2006}.

	\item \textbf{Inverse Cumulative Distribution Function (CDF) Transformation} \label{eq:inverse_cdf}: Though computationally more demanding, this method has the flexibility to handle any distribution with a known inverse CDF. It transforms a uniform random variable into a GRV through the inverse of the CDF \cite{Yahav2007}.
\end{enumerate}

\subsection{The Marsaglia Polar Method Explained}
\label{sec:classicalgauss}
GRVs play a pivotal role in the forward phase of SD models, serving as the source of Gaussian noise introduced into images. These variables find extensive utility not only in machine learning but also in the physical sciences. A standard GRV, denoted \( Z \), has a mean \( \mu \) of zero and a variance \( \sigma^2 \) of one. Its Probability Density Function (PDF) and Cumulative Distribution Function (CDF) are described as follows:

\begin{align}
	\phi(z) &= \frac{1}{\sqrt{2 \pi}} \exp \left(-\frac{1}{2} z^2\right), \label{eq:pdf} \\
	\Phi(z) &= \mathbb{P}[Z < z] = \int_{-\infty}^{z} \phi(s) \, \mathrm{d} s. \label{eq:cdf}
\end{align}

The Marsaglia method, which can be viewed in Algorithm \ref{alg::marsaglia}, commences by generating two uniform random variables, \( u1 \) and \( u2 \). These are subsequently transformed to lie within the interval \([-1, 1]\), creating \( u \) and \( v \). Here, \( u \) and \( v \) act as Cartesian coordinates on a two-dimensional plane. The algorithm's crux is to ascertain if the point \((u, v)\) resides within a unit circle. This is determined by calculating \( s = u^2 + v^2 \) and ensuring that \( s < 1 \) and \( s \neq 0 \). Should the point fall outside the unit circle, a new set of \( u \) and \( v \) is generated. Once a valid point is identified, a scaling factor is computed using \( s \), and this factor is applied to \( u \) and \( v \) to produce \( z1 \) and \( z2 \), two normally distributed variables. The algorithm then outputs \( z2 \). Herein, we offer an exhaustive account of the Marsaglia method's inner workings.

\begin{algorithm}
	\caption{The Marsaglia Polar Method for Producing Normally Distributed Random Variables}
	\begin{algorithmic}[1]
		\State Initialise \( u1, u2, u, v, s, z1, z2 \)
		\While{True}
		\State \( u1, u2 \sim U(0, 1) \) \Comment{Quantum uniform random number generator}
		\State \( u = 2u1 - 1 \)
		\State \( v = 2u2 - 1 \)
		\State \( s = u^2 + v^2 \)
		\If{\( s \geq 1 \) or \( s = 0 \)}
		\State Continue \Comment{Revert to step one}
		\Else
		\State \( z1 = u \sqrt{-2 \ln s / s} \)
		\State \( z2 = v \sqrt{-2 \ln s / s} \)
		\State \textbf{break}
		\EndIf
		\EndWhile
		\State \Return \( z2 \)
	\end{algorithmic}
\label{alg::marsaglia}
\end{algorithm}

\section{Quantum Random Number Generators (QRNGs)}
\label{sec:qrng}

While the prevailing literature often emphasises Parametric Quantum Circuits (PQCs) \cite{mcclean2016theory, farhi2014quantum} and variational quantum eigensolvers (VQEs) \cite{peruzzo2014variational, kandala2017hardware}, our approach diverges by adopting a non-parametric methodology, obviating the need for hybrid optimisation strategies. For completeness, we first expound upon the parametric methods.

\subsection{Parametric Quantum Circuits (PQCs)}
\label{sec:pqc}
The quantum circuitry employed for generating GRVs has garnered significant attention, particularly within the context of Quantum Circuit Born Machines (QCBMs) \cite{cheng2018quantum}. A recurring theme in current methodologies is the integration of quantum and classical systems. Romero and Aspuru-Guzik \cite{romero2019variational}, for instance, employed a variational quantum generator within a Generative Adversarial Network (GAN), with the quantum circuit functioning as the generator and a classical neural network acting as the discriminator. Similarly, works by Gili et al. \cite{gili2023generative} and Liu and Wang \cite{liu2018differentiable} have further enriched the landscape of quantum-classical hybrid systems in the context of probability distribution learning and optimisation.

\subsection{Non-Parametric Quantum Circuits}
\label{sec:non-parametric}
Our methodology, depicted in Figure~\ref{fig::our_pipeline}, centres around non-parametric quantum circuits \cite{biamonte2017,nielsen_chuang_2010,benedetti2019,qsvt}. It specifically aims at transforming uniform distributions to Gaussian noise, with the details elaborated upon subsequently. The quantum-enhanced implementation of our algorithm, delineated in Code snippet~\ref{lst:quantum_marsaglia}, is modelled on Algorithm~\ref{alg::marsaglia}.

\begin{lstlisting}[language=Python, caption={Quantum-Enhanced Marsaglia Algorithm for Gaussian Noise Generation}, label=lst:quantum_marsaglia, captionpos=b]
	while True:
		u1 = (self.q_uniform.Q_sample())
		u2 = (self.q_uniform.Q_sample())
		u = 2 * u1 - 1
		v = 2 * u2 - 1
		s = u ** 2 + v ** 2
	if s < 1 and s != 0:
		break
	factor = (-2 * np.log(s) / s) ** 0.5
	z1 = u * factor
	z2 = v * factor
	return torch.tensor(z1, dtype=float), torch.tensor(z2, dtype=float)
\end{lstlisting}

\subsection{Generation of Uniform Random Numbers}
\label{sec:disc-uniform}

Figure~\ref{fig:static_pqc_ansatz_sampling} portrays a quantum circuit designed for generating random rotations. These rotations, functioning as non-trainable parameters, are subsequently utilised in a separate ansatz to generate uniform distributions. The quantum noise emanating from these circuits is employed in image corruption tasks and the simulation of diffusion processes akin to Brownian motion.

\begin{figure}[H]
	\begin{minipage}[b]{.21\textwidth}
		\centering
		\includegraphics[width=\textwidth]{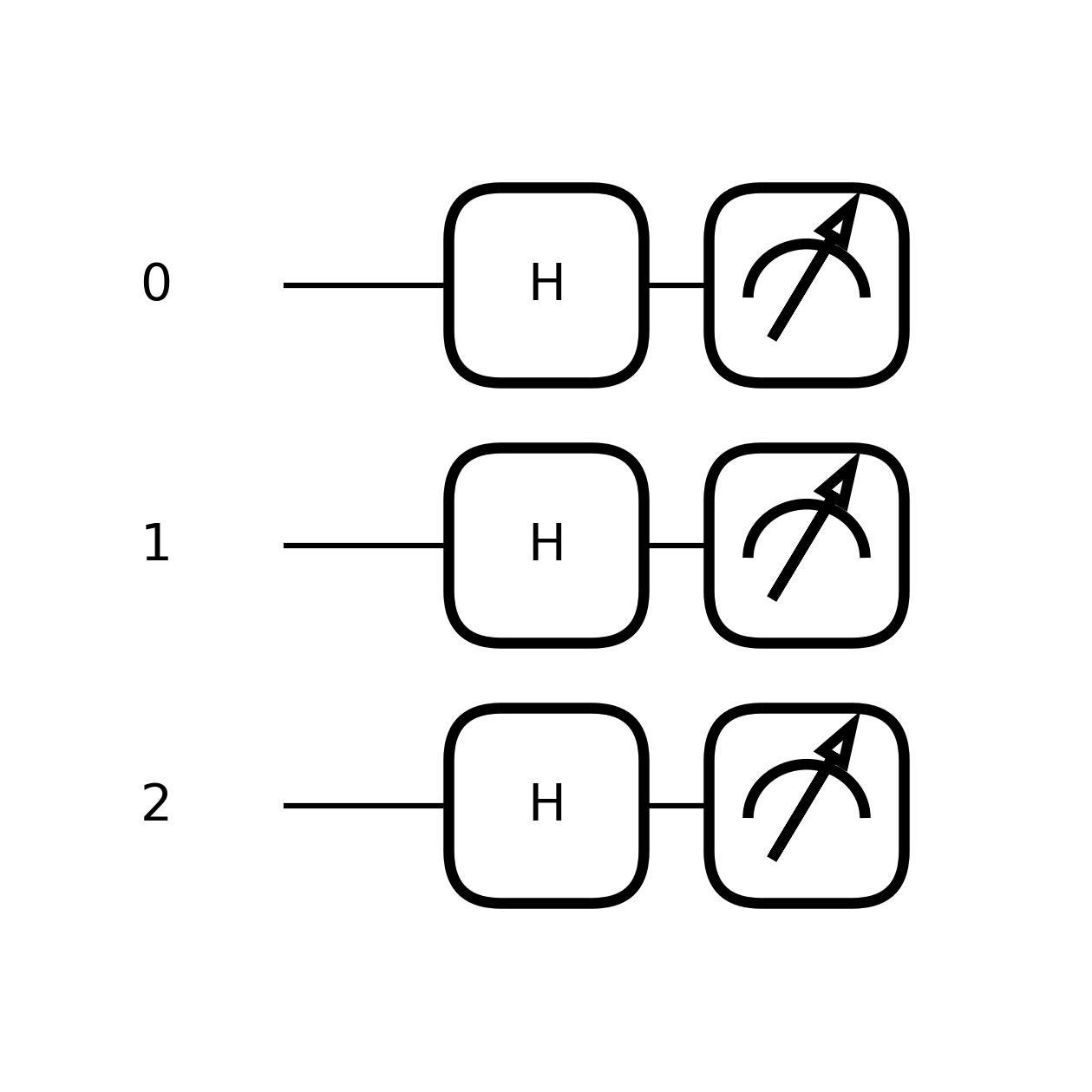}
		(a)
	\end{minipage}
	\begin{minipage}[b]{.45\textwidth}
		\begin{lstlisting}[caption={},language=Python]
			def Q_PQC_ROT_ANZATS(n_qubits):
			for i in range(n_qubits):
			qml.Hadamard(wires=i)
			_expectations = [qml.sample(qml.PauliZ(i)) for i in range(n_qubits)]
			return _expectations
		\end{lstlisting}
		(b)
	\end{minipage}
	\caption[A Simplified Circuit Ansatz for Generating Random Rotations]{A Simplified Circuit Ansatz for Generating Random Rotations: (a) Circuit Ansatz visualisation, (b) Corresponding Python code snippet.}
	\label{fig:static_pqc_ansatz_sampling}
\end{figure}

\subsection{Sampling vs. Measuring Expectation Values}
\label{sec:sample_exp}

In quantum computing, \textbf{sampling and calculating expectation values} are distinct yet interconnected concepts. Much like expectation values, sampling plays a crucial role in quantum algorithms and simulations. While \texttt{qml.expval()} computes expectation values without collapsing the wavefunction, \texttt{qml.sample()} performs a measurement that results in wavefunction collapse, yielding a specific eigenvalue of the measured operator \cite{bergholm2018pennylane}.
\begin{definition}
	(Expectation Values) The expectation value of an operator \(\hat{O}\) for a state \(|\Psi\rangle\) is the average value of the observable corresponding to \(\hat{O}\), given by
	\begin{equation}
		\langle \hat{O} \rangle = \langle \Psi|\hat{O}|\Psi\rangle.
	\end{equation}
\end{definition}
\begin{example}
	For a qubit in the state \(|\Psi\rangle = \alpha|\uparrow\rangle + \beta|\downarrow\rangle\), the expectation value of the Pauli-Z operator can be calculated as follows:
	\begin{align}
		\langle \hat{\sigma}_z \rangle &= \langle \Psi|\hat{\sigma}_z|\Psi\rangle, \\
		&= \left(\alpha^*\langle\uparrow| + \beta^*\langle\downarrow|\right)\hat{\sigma}_z\left(\alpha|\uparrow\rangle + \beta|\downarrow\rangle\right), \\
		&= \alpha^*\alpha\langle\uparrow|\hat{\sigma}_z|\uparrow\rangle + \beta^*\beta\langle\downarrow|\hat{\sigma}_z|\downarrow\rangle, \\
		&= |\alpha|^2\cdot 1 + |\beta|^2\cdot (-1), \\
		&= |\alpha|^2 - |\beta|^2.
	\end{align}
\end{example}
\begin{example}
	Consider the state
	\begin{equation}
		|\psi\rangle=\frac{3}{7}|00\rangle+\frac{6}{7}|01\rangle+\frac{2}{7}|10\rangle,
	\end{equation}
	and the operator \(I \otimes \sigma_z\), representing the \(\sigma_z\) operator acting on qubit B. The expectation value of \(\sigma_z\) for qubit B can be calculated as:
	\begin{align}
		\langle\sigma_z\rangle_B &= \langle\psi|\left(I \otimes \sigma_z\right)|\psi\rangle \\
		&= \left(\frac{3}{7}\right)^2\cdot 1 + \left(\frac{6}{7}\right)^2\cdot (-1) + \left(\frac{2}{7}\right)^2\cdot 1 \\
		&= \frac{9}{49} - \frac{36}{49} + \frac{4}{49} \\
		&= \frac{-23}{49}
	\end{align}
\end{example}
\begin{definition}
	(Sampling) Sampling in the context of a quantum state \(|\Psi\rangle\) refers to the process of performing a measurement on the state in a specific basis, such as the Pauli-Z basis. The outcome is probabilistic, and the state collapses to one of the eigenstates of the measured operator. Mathematically, if \(|\phi_i\rangle\) are the eigenstates of an operator \(\hat{O}\) with eigenvalues \(o_i\), then sampling \(|\Psi\rangle\) with respect to \(\hat{O}\) yields \(o_i\) with probability \(|\langle \phi_i|\Psi\rangle|^2\).
\end{definition}

\begin{example}
	Consider a qubit in the state \(|\Psi\rangle = \frac{1}{\sqrt{2}}(|\uparrow\rangle + |\downarrow\rangle)\). The probabilities for each outcome when sampling this state in the Pauli-Z basis can be calculated as follows:
	\begin{align}
		P(+1) &= \left|\langle \uparrow|\Psi\rangle\right|^2 \\
		&= \left|\frac{1}{\sqrt{2}}\right|^2 \\
		&= \frac{1}{2}, \\
		P(-1) &= \left|\langle \downarrow|\Psi\rangle\right|^2 \\
		&= \left|\frac{1}{\sqrt{2}}\right|^2 \\
		&= \frac{1}{2}.
	\end{align}
	Therefore, sampling would yield \(+1\) with a probability of \(0.5\) and \(-1\) with a probability of \(0.5\).
\end{example}
\begin{example}
	For a two-qubit state
	\begin{equation}
		|\psi\rangle = \frac{1}{\sqrt{3}}(|00\rangle + |01\rangle + |10\rangle),
	\end{equation}
	sampling the first qubit in the Pauli-Z basis would yield the following probabilities:
	\begin{align}
		P(+1) &= \left|\frac{1}{\sqrt{3}}\right|^2 + \left|\frac{1}{\sqrt{3}}\right|^2 \\
		&= \frac{1}{3} + \frac{1}{3} \\
		&= \frac{2}{3}, \\
		P(-1) &= \left|\frac{1}{\sqrt{3}}\right|^2 \\
		&= \frac{1}{3}.
	\end{align}
	Hence, \(+1\) would be obtained with a probability of \(\frac{2}{3}\) and \(-1\) with a probability of \(\frac{1}{3}\).
\end{example}
We are \textbf{concerned with sampling here}. The primary distinction between the expressions \texttt{[qml.sample(qml.PauliZ(i)) ]} and \texttt{[qml.expval(qml.PauliZ(i))]} resides in the nature of the result they yield \cite{bergholm2018pennylane} . The former returns a sampled measurement outcome, while the latter procures the expectation value. \texttt{qml.sample()} executes a measurement on the qubit in the Pauli-Z basis, causing the collapse of the wavefunction and \textbf{randomly returning either 0 or 1} based on the probabilities inherent to the qubit state. On the other hand, \texttt{qml.expval()} computes the expectation value \(\langle Z \rangle\) of the Pauli-Z operator on the qubit, \textbf{yielding the average value} we would anticipate measuring, without inducing the collapse of the wavefunction. In essence, \texttt{qml.sample()} provides a random measurement sample (0 or 1) while \texttt{qml.expval()} offers the expected average value of a measurement (a value between 0 and 1). \textbf{Sampling induces the collapse of the state, while expectation values permit further quantum processing. }
\begin{figure}[h!]
	\centering
	\begin{minipage}{.35\textwidth}
		\centering
		\includegraphics[width=\linewidth]{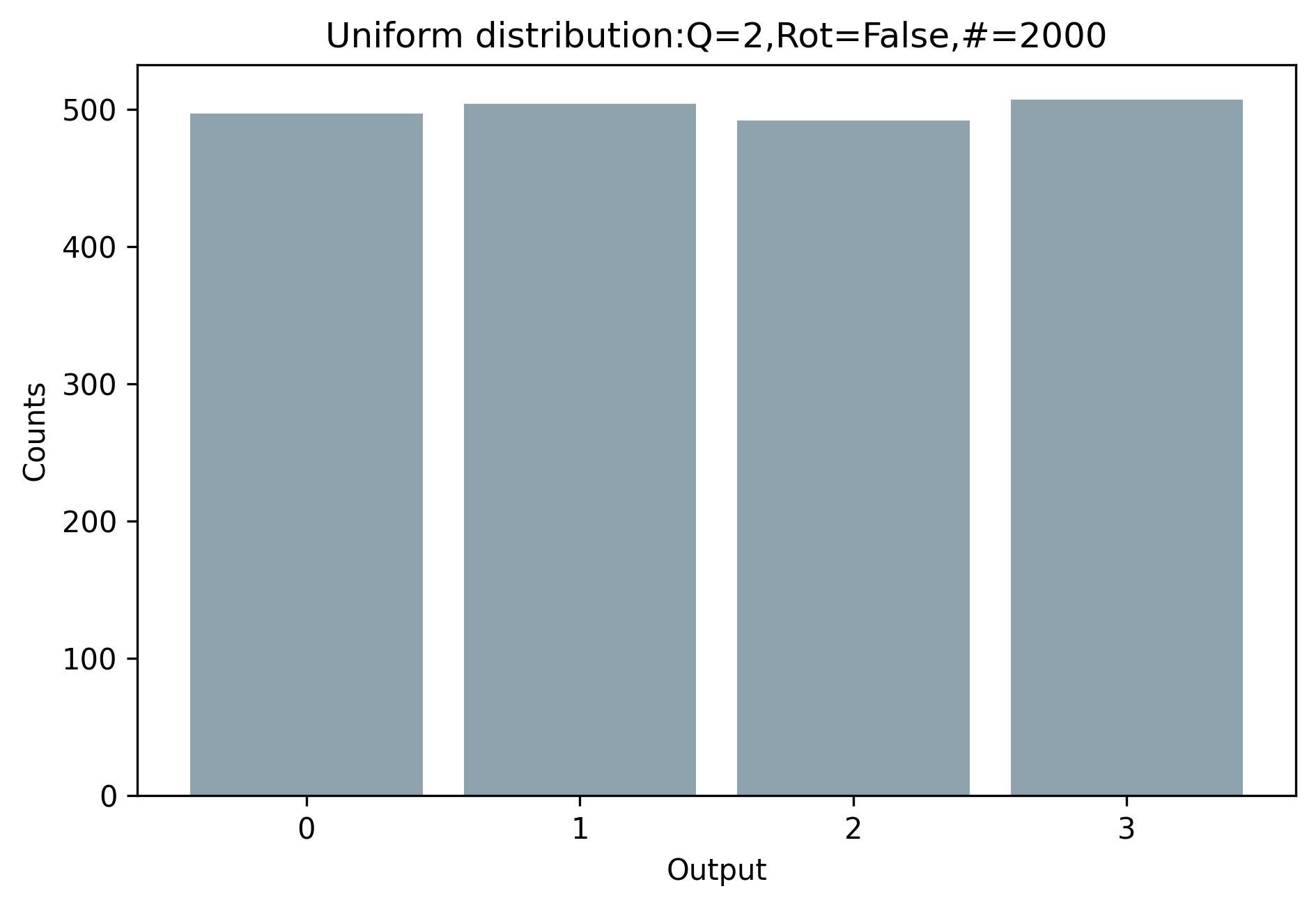}
	\end{minipage}%
	\begin{minipage}{.35\textwidth}
		\centering
		\includegraphics[width=\linewidth]{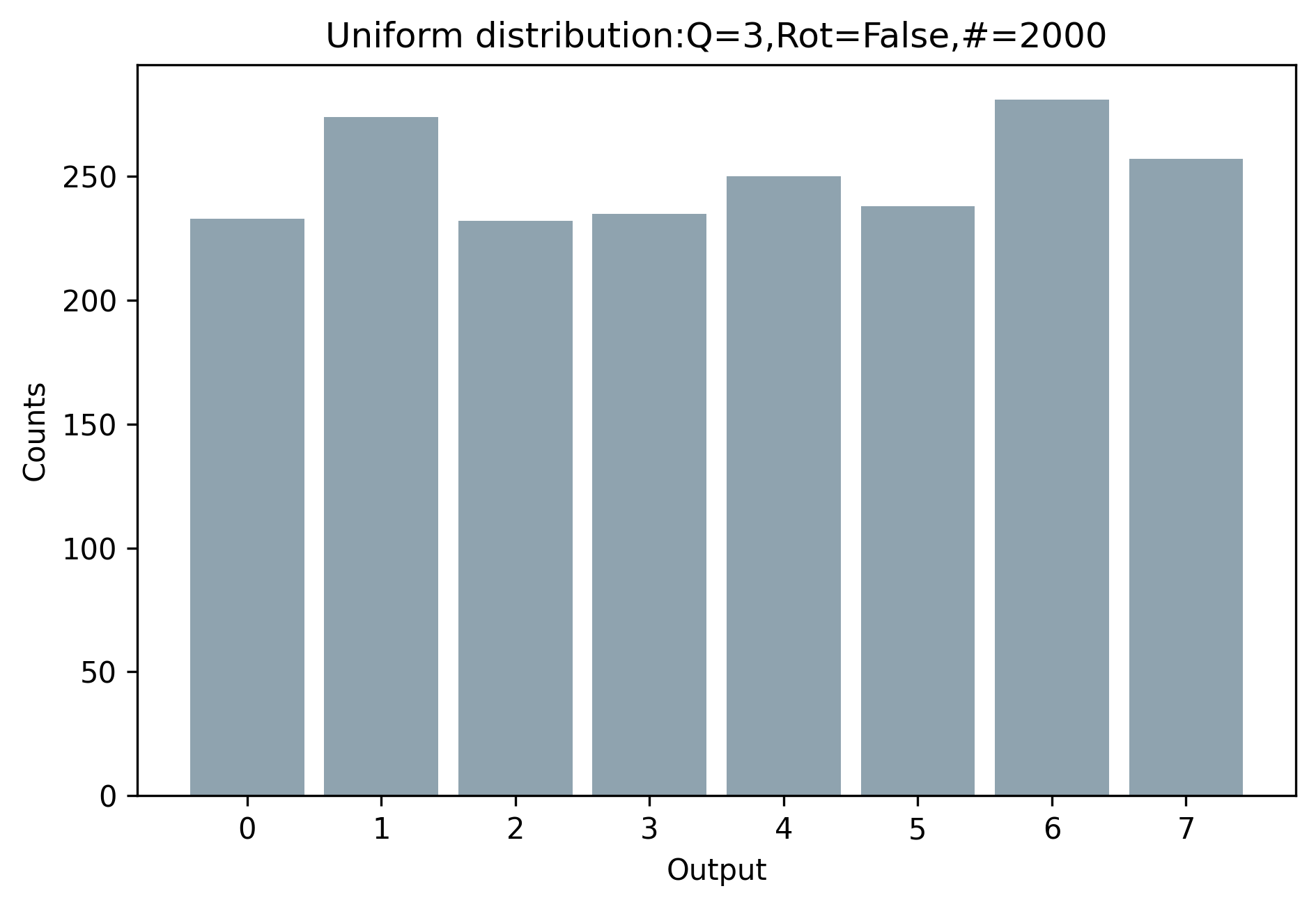}
	\end{minipage}

	\begin{minipage}{.75\textwidth}
		\centering
		\includegraphics[width=\linewidth]{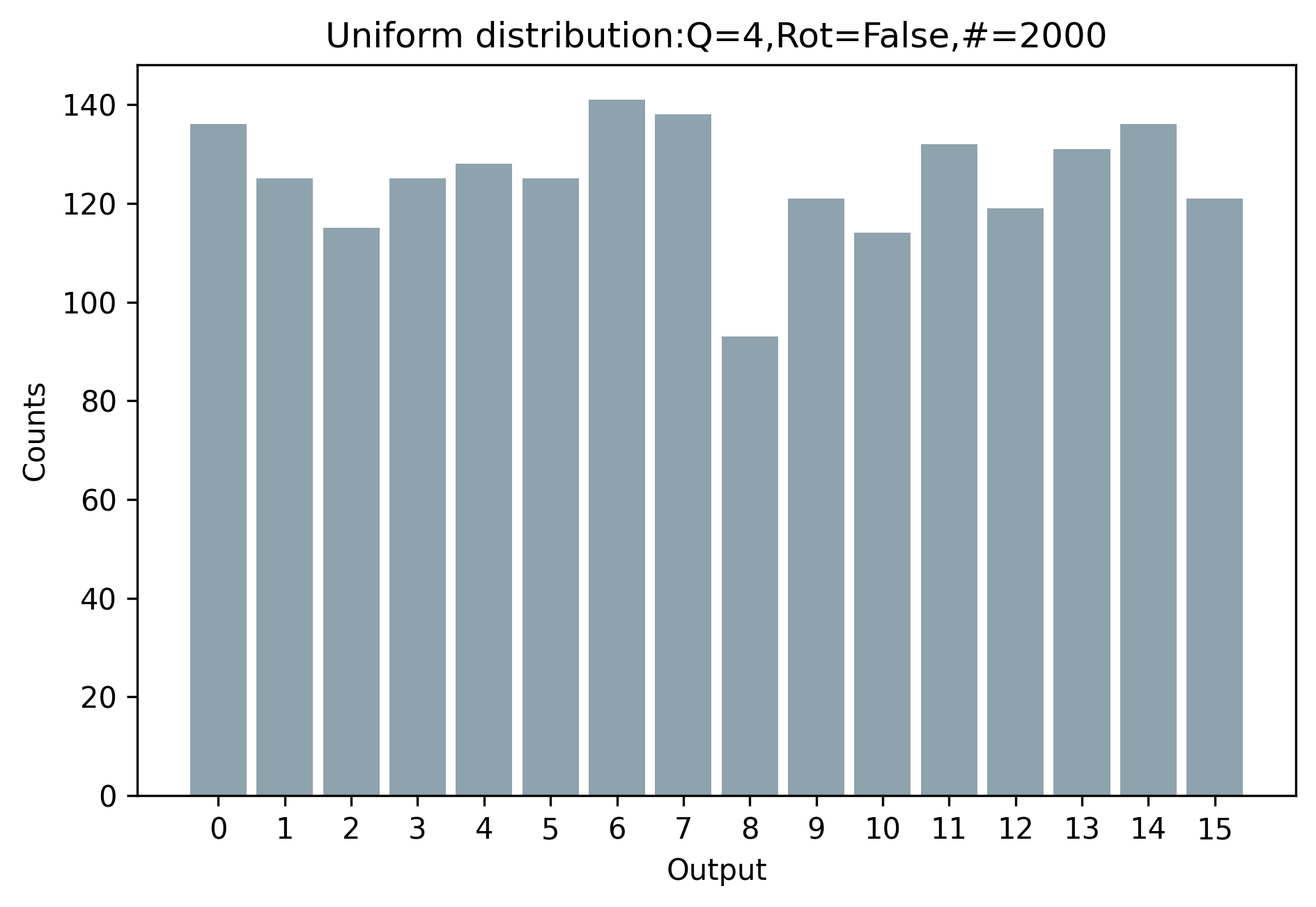}
	\end{minipage}%
	\caption[Plots of Random Rotation Angles Generated by a Quantum Circuit]{Plots of random rotation angles generated by a quantum circuit. The circuit applies a Hadamard gate to each qubit, creating a superposition of states. Each state is then measured in the Pauli-Z basis, resulting in a series of 0s and 1s (binary outcomes). These binary outcomes are converted to a binary index value, which is then normalised to a value between 0 and 1 by dividing by the total number of possible states (the space span). Each plot shows the distribution of these generated values for different numbers of qubits. The limited range of qubit string bits is due to the finite number of binary outcomes that can be generated by a given number of qubits.}
	\label{fig::uniform_qubit}
\end{figure}
The Python method \texttt{toBinaryIndex} performs the task of converting the sampling outcomes to a binary index value. It takes a list of qubit measurements as input and returns an integer value representing the binary index.
\begin{align}
	\text{value} &= \sum_{i=0}^{n-1} 2^i \times \text{sample}_i \label{eq:binary_index}
\end{align}
The binary index is then normalised to a value between 0 and 1 by dividing it by the total number of possible states, effectively spanning the state space.
\begin{align}
	\text{Normalised Value} &= \frac{\text{value}}{2^n - 1} \label{eq:normalisation}
\end{align}
The distribution of these normalised values is plotted in Figure~\ref{fig::uniform_qubit} for varying numbers of qubits. The limited range of qubit string bits is a consequence of the finite number of binary outcomes that can be generated by a given number of qubits.

\subsection{Generating Uniformly Distributed Rotations}
\label{sec:rots}
As previously highlighted, a central premise of our manuscript is the cascading of two quantum circuits, wherein the first circuit supplies the second with random rotations, thereby introducing an additional layer of stochasticity. Quantum rotations serve as the cornerstone for myriad quantum manipulations and are essential for altering quantum states. Within this article, we briefly examine the mathematical formulations of quantum rotations, specifically concentrating on rotations about the $x, y$, and $z$ axes, symbolised as $R_x, R_y$, and $R_z$ respectively.
\begin{figure}[H]
	\centering
	\begin{minipage}{.35\textwidth}
		\centering
		\includegraphics[width=\linewidth]{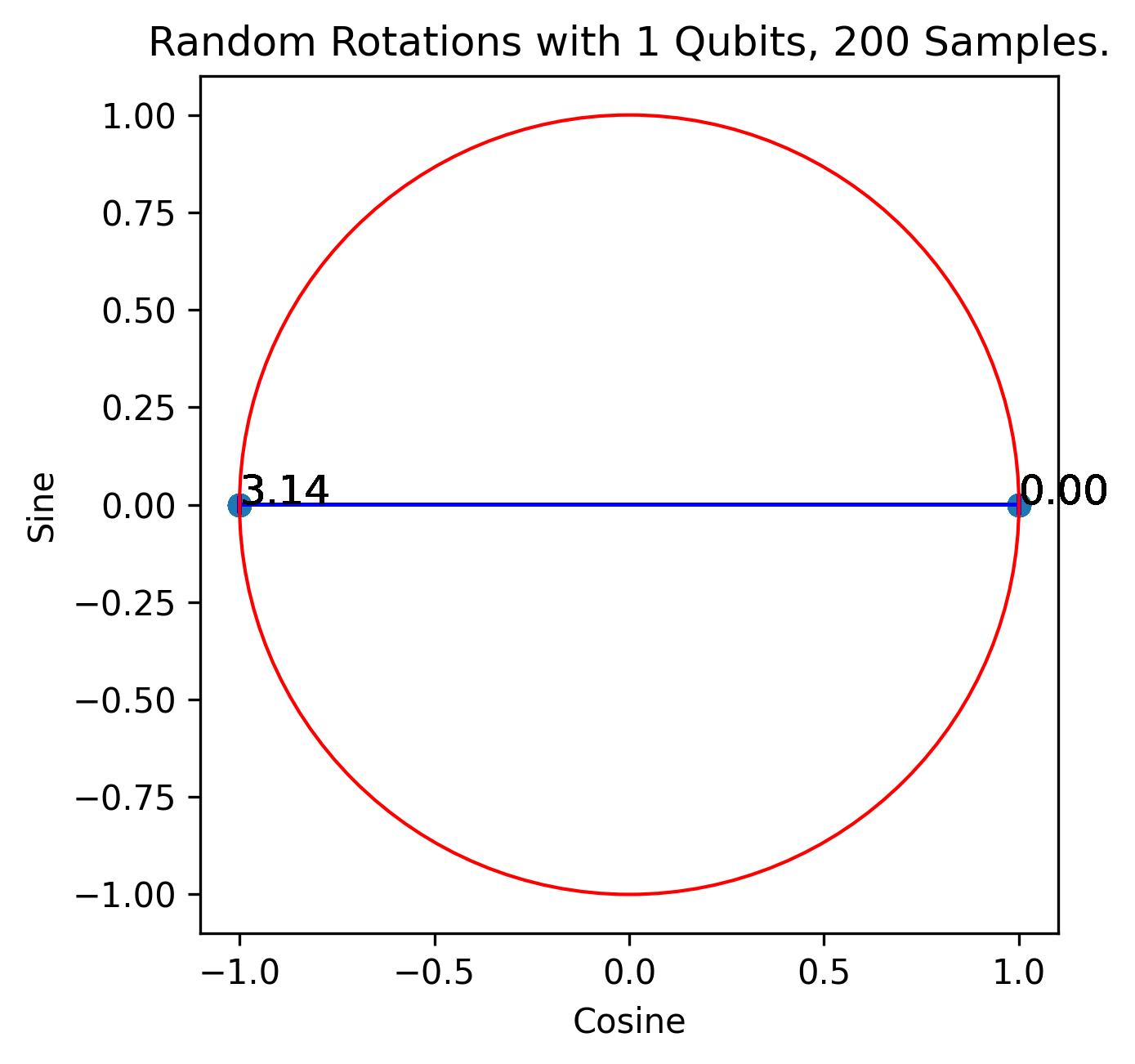}
	\end{minipage}%
	\begin{minipage}{.35\textwidth}
		\centering
		\includegraphics[width=\linewidth]{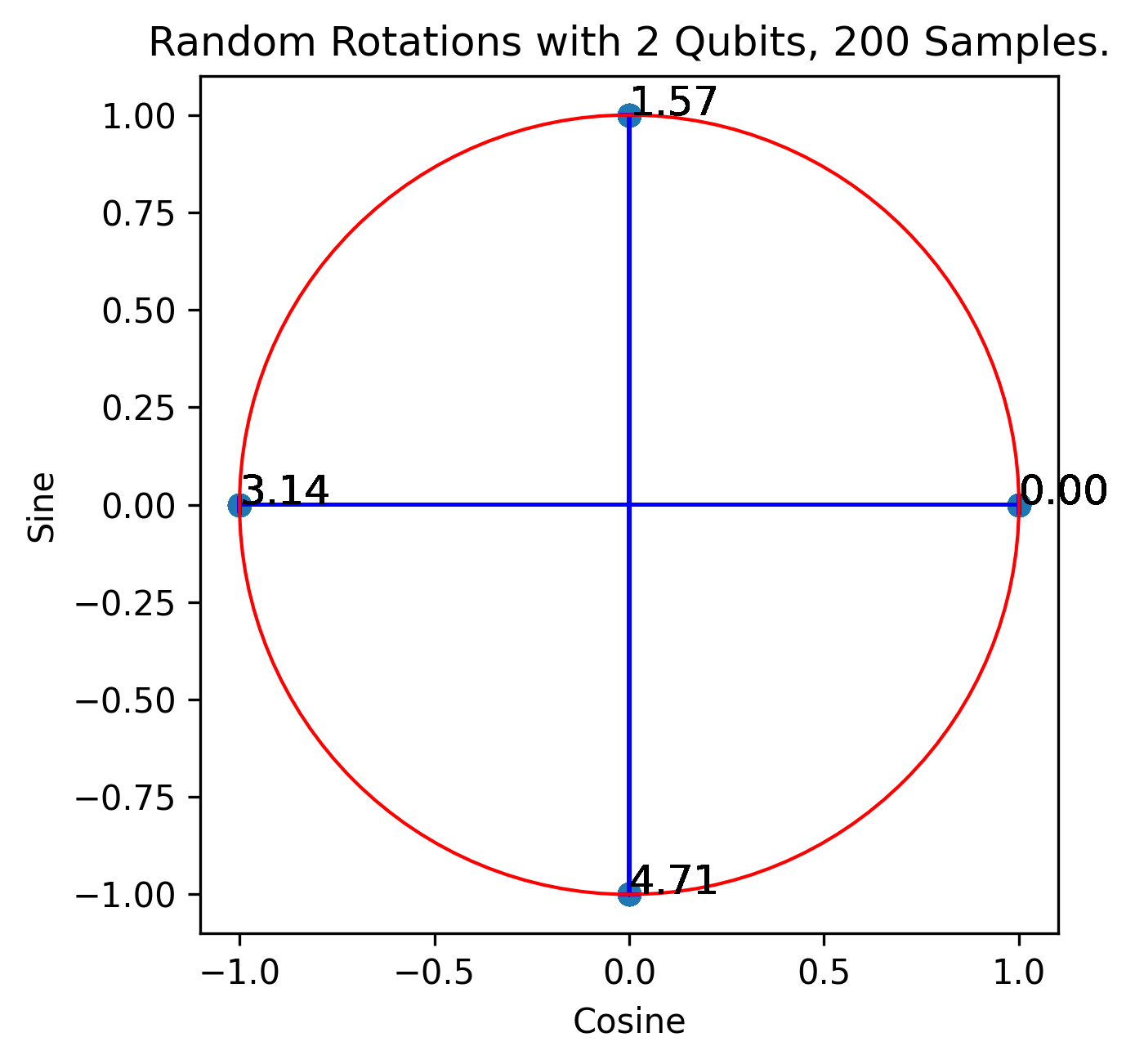}
	\end{minipage}

	\begin{minipage}{.35\textwidth}
		\centering
		\includegraphics[width=\linewidth]{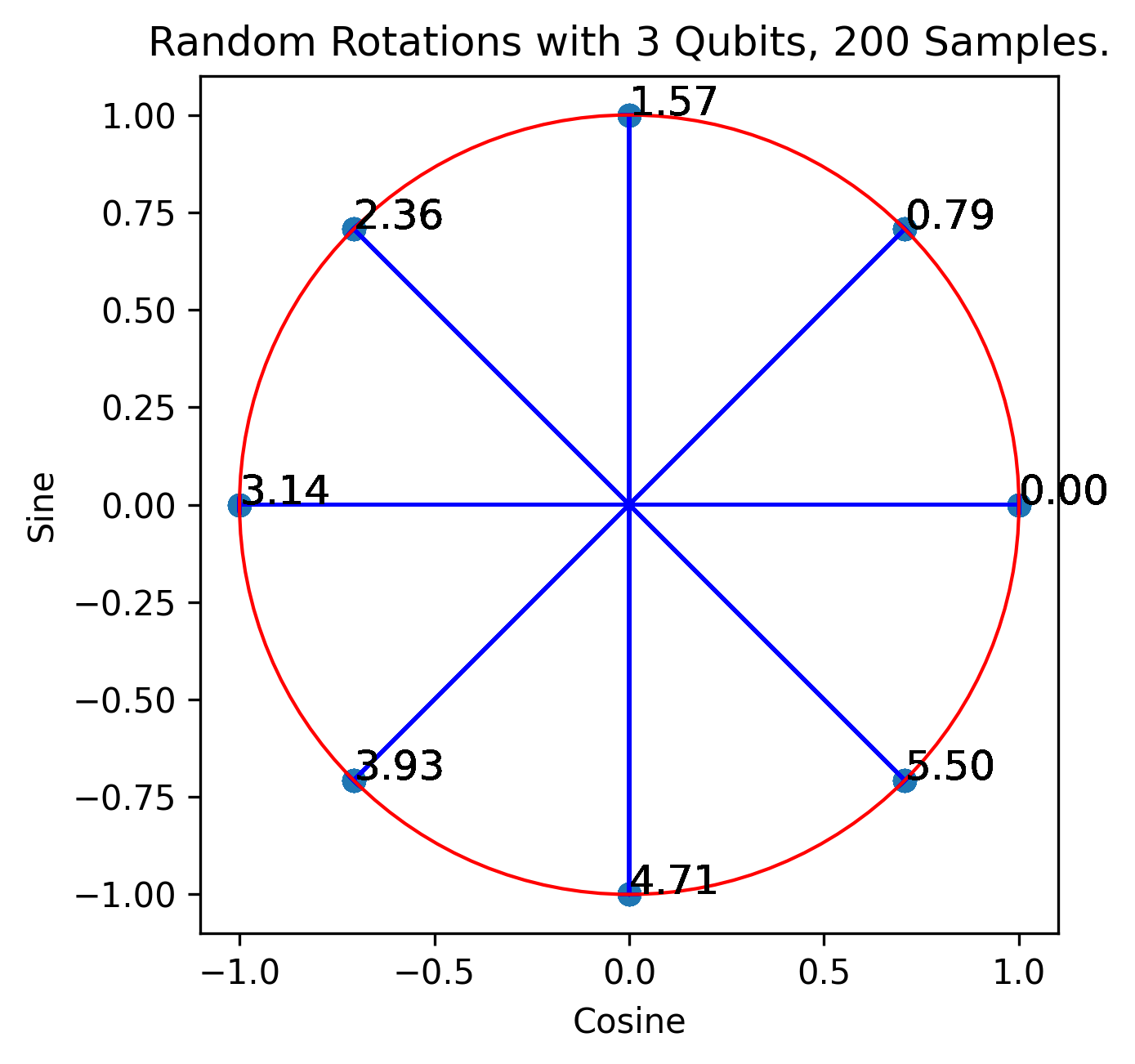}
	\end{minipage}%
	\begin{minipage}{.35\textwidth}
		\centering
		\includegraphics[width=\linewidth]{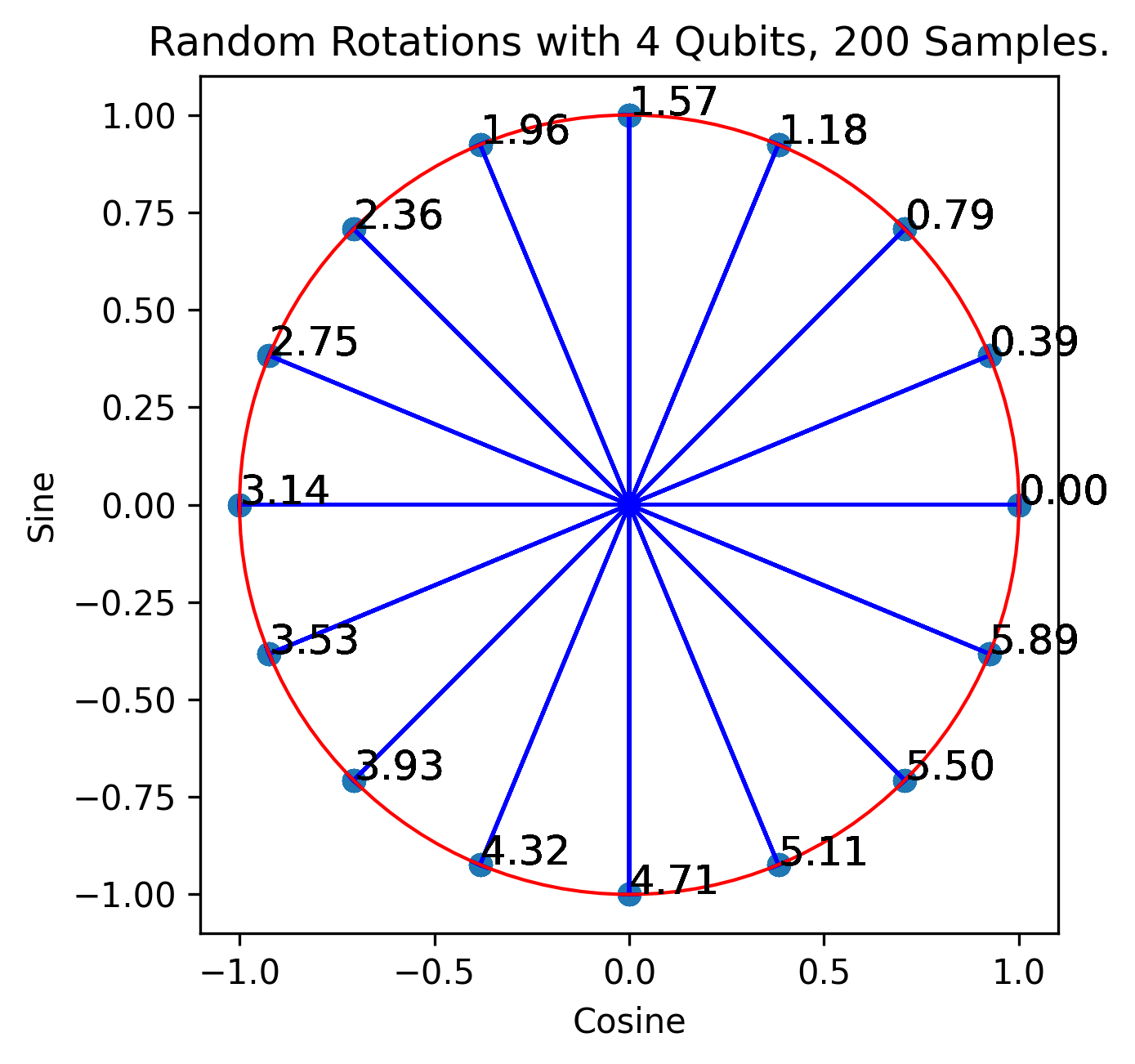}
	\end{minipage}
	\caption[Plots of Quantum Sampling Measurement Results on a Unit Circle]{Plots of random rotation angles generated by a quantum circuit. The circuit employs a Hadamard gate on each qubit to create a superposition of quantum states. Following state preparation, each qubit state is measured in the Pauli-Z basis, yielding a series of binary outcomes, represented as 0s and 1s. These outcomes are then converted into a single binary index value, which is normalized to a value between 0 and 1 by dividing by the total number of possible states. This normalized value is subsequently scaled to a rotation angle between 0 and \(2\pi\) radians. Each plot illustrates the distribution of these generated angles for varying numbers of qubits, depicted in polar coordinates with the angle represented by the position on the circle. The limited range of angles is a consequence of the finite number of binary outcomes that can be generated by a given number of qubits.}
	\label{fig:random_rotations}
\end{figure}
\begin{definition}
	(Quantum Rotation Operators) The quantum rotation operators \(R_x\), \(R_y\), and \(R_z\) are defined as:
	\begin{gather}
		R_x(\theta) = e^{-i \theta X/2} = \cos\left(\frac{\theta}{2}\right) I - i \sin\left(\frac{\theta}{2}\right) X, \\
		R_y(\theta) = e^{-i \theta Y/2} = \cos\left(\frac{\theta}{2}\right) I - i \sin\left(\frac{\theta}{2}\right) Y, \\
		R_z(\theta) = e^{-i \theta Z/2} = \cos\left(\frac{\theta}{2}\right) I - i \sin\left(\frac{\theta}{2}\right) Z,
	\end{gather}
	where \(X\), \(Y\), and \(Z\) are the Pauli matrices, \(I\) is the identity matrix, and \(\theta\) is the angle of rotation.
\end{definition}

\begin{example}
	The matrix representations of the quantum rotation operators \(R_x\), \(R_y\), and \(R_z\) are given by:
	\begin{align}
		R_x(\theta) &= \left[\begin{array}{cc}\cos \frac{\theta}{2} & -i \sin \frac{\theta}{2} \\ -i \sin \frac{\theta}{2} & \cos \frac{\theta}{2}\end{array}\right], \\
		R_y(\theta) &= \left[\begin{array}{cc}\cos \frac{\theta}{2} & -\sin \frac{\theta}{2} \\ \sin \frac{\theta}{2} & \cos \frac{\theta}{2}\end{array}\right], \\
		R_z(\theta) &= \left[\begin{array}{cc}e^{-i \frac{\theta}{2}} & 0 \\ 0 & e^{i \frac{\theta}{2}}\end{array}\right].
	\end{align}
\end{example}
\index{Python}
\begin{example}
	We consider a three-qubit quantum circuit where each qubit undergoes a Hadamard transformation followed by a general rotation characterized by three parameters: \( \theta, \phi, \lambda \).
	\begin{enumerate}
		\item \textbf{Initial State:} The initial state of the three qubits is prepared by another quantum circuit consisting solely of Hadamard gates. Mathematically, this is represented as:
		\begin{equation}
			|\psi_{\text{initial}}\rangle = (H \otimes H \otimes H) |0\rangle^{\otimes 3}.
		\end{equation}

		\item \textbf{Total Operation:} The total unitary operation acting on the initial state consists of Hadamard gates followed by general rotations on each qubit. This is given by:
		\begin{equation}
			U_{\text{total}} = (U(\theta_1, \phi_1, \lambda_1) \otimes U(\theta_2, \phi_2, \lambda_2) \otimes U(\theta_3, \phi_3, \lambda_3)) \cdot (H \otimes H \otimes H) \cdot (H \otimes H \otimes H).
		\end{equation}

		\item \textbf{Final State:} The final state of the system is then:
		\begin{equation}
			|\psi_{\text{final}}\rangle = U_{\text{total}} |\psi_{\text{initial}}\rangle.
		\end{equation}
	\end{enumerate}

	Let's consider a numerical example for the system. We'll use the following rotation angles for each qubit:
\begin{tabular}{lll}
	\( \theta_1 = \frac{\pi}{4} \) & \( \phi_1 = \frac{\pi}{2} \) & \( \lambda_1 = \pi \) \\
	\( \theta_2 = \frac{\pi}{3} \) & \( \phi_2 = \frac{\pi}{4} \) & \( \lambda_2 = \frac{\pi}{2} \) \\
	\( \theta_3 = \frac{\pi}{6} \) & \( \phi_3 = \frac{\pi}{3} \) & \( \lambda_3 = \frac{\pi}{4} \)
\end{tabular}
	The Hadamard gate \( H \) is represented as:
	\begin{equation}
		H = \frac{1}{\sqrt{2}} \begin{pmatrix}
			1 & 1 \\
			1 & -1
		\end{pmatrix}
	\end{equation}

	And a general rotation \( U(\theta, \phi, \lambda) \) is given by:
	\begin{equation}
		U(\theta, \phi, \lambda) = \begin{pmatrix}
			\cos(\frac{\theta}{2}) & -e^{i \lambda} \sin(\frac{\theta}{2}) \\
			e^{i \phi} \sin(\frac{\theta}{2}) & e^{i (\phi + \lambda)} \cos(\frac{\theta}{2})
		\end{pmatrix}
	\end{equation}

	First, we'll find the initial state \( |\psi_{\text{initial}}\rangle \) after applying the Hadamard gates:
	\begin{equation}
		|\psi_{\text{initial}}\rangle = \frac{1}{\sqrt{8}}(|000\rangle + |001\rangle + |010\rangle + |011\rangle + |100\rangle + |101\rangle + |110\rangle + |111\rangle)
	\end{equation}

	Next, we'll find the unitary operation \( U_{\text{total}} \) for the Hadamard and rotation gates:
	\begin{equation}
		U_{\text{total}} = (U(\theta_1, \phi_1, \lambda_1) \otimes U(\theta_2, \phi_2, \lambda_2) \otimes U(\theta_3, \phi_3, \lambda_3)) \times (H \otimes H \otimes H)
	\end{equation}

	Finally, we'll find the final state \( |\psi_{\text{final}}\rangle \):
	\begin{equation}
		|\psi_{\text{final}}\rangle = U_{\text{total}} |\psi_{\text{initial}}\rangle
	\end{equation}

	The evaluation of \( |\psi_{\text{final}}\rangle \) is executed solely via numerical software, notably the PennyLane quantum simulation library.

\end{example}
\subsection{Generating GRV's}
\label{sec:gauss}
In the present investigation, we solely employ the Marsaglia  \cite{Vigna2014b} polar algorithm (\ref{alg::marsaglia}) as a mechanism for generating normally distributed random numbers. While the extant literature presents a variety of enhancements or alternatives to the Marsaglia polar method, such as the one mentioned above~(\ref{eq:box_muller}), we have chosen not to explore these avenues, given the empirical efficacy of the method under consideration. It is crucial to underscore that while various studies propose potential refinements to the method, the selection of an optimal approach is intrinsically tied to the specific requirements of the application in focus.
The prospect of harnessing the uncertainty of quantum measurements for generative modelling presents an intriguing research direction \cite{romero2019variational,liu2018differentiable}. The premise, posits that it is feasible to iteratively superimpose tiny amounts of quantum noise Figure~ \ref{fig:quantum_gaussian_samples}, denoted as $\epsilon$, onto any image $X$ over $t$ timesteps, thereby transmuting $X$ into a pure Gaussian noise sample $T$. In a reciprocal manner, given that $t$ is adequately large the noise-laden sample $T$ can be reverted to the original, noise-free image $X$ through the systematic elimination of the superimposed noise.
\begin{figure}[h!]
	\centering
	\includegraphics[width=\linewidth]{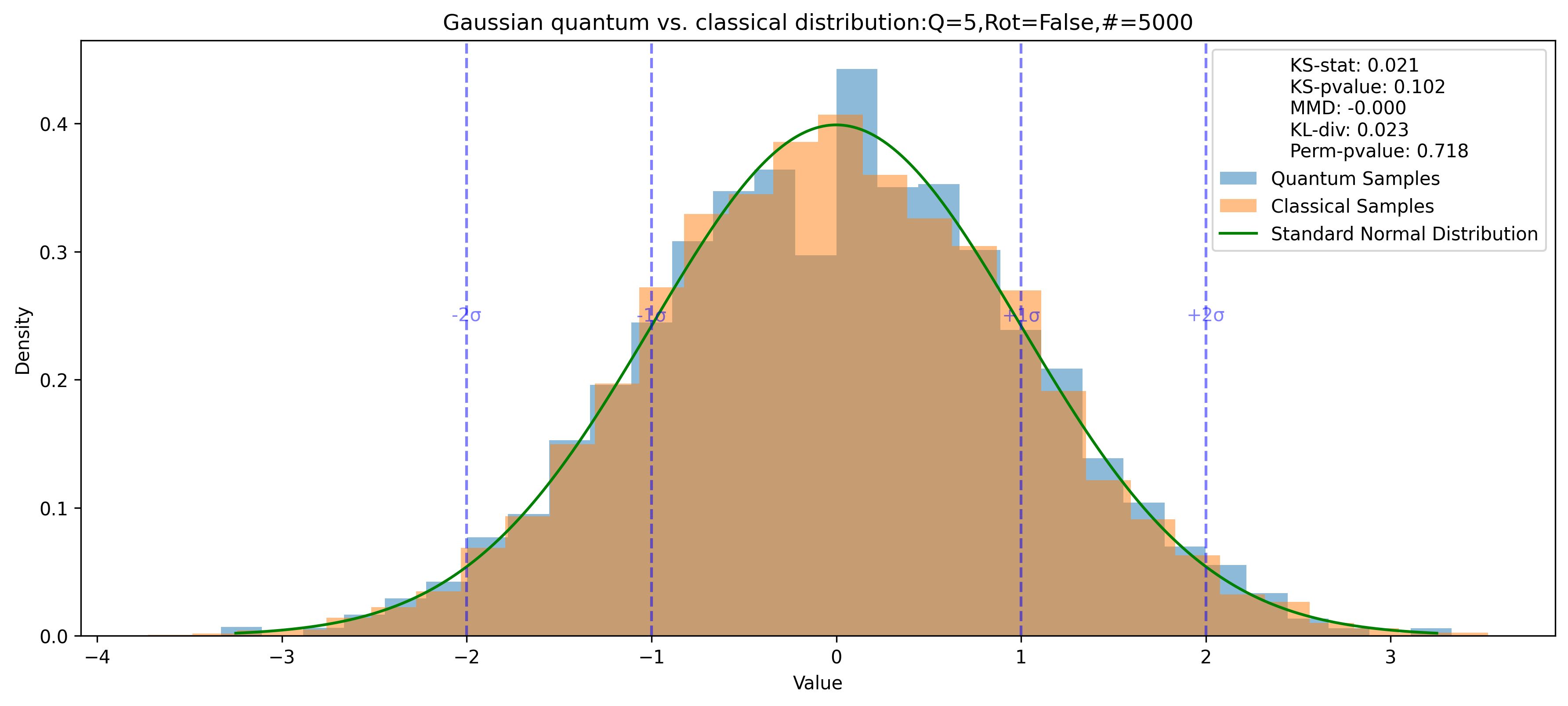}
	\caption[Plots of 2000 Gaussian samples from a quantum circuit]{This figure depicts the probability density function (PDF) of a standard normal (Gaussian) distribution, generated from 2000 samples of a quantum circuit. The red dashed lines mark positions of one, two, and three standard deviations away from the mean, in both the positive and negative directions \cite{brent2010fast}. The labels '-1\(\sigma\)', '+1\(\sigma\)', '-2\(\sigma\)', '+2\(\sigma\)', '-3\(\sigma\)', '+3\(\sigma\)' represent points one, two, and three standard deviations away from the mean, respectively. This visualisation assists in understanding the empirical rule for a standard normal distribution, also known as the 68-95-99.7 rule.}
	\label{fig:quantum_gaussian_samples}
\end{figure}

In our software stack QonFusion, the \texttt{QuantumRandomGenerator} class is designed to generate a uniform random distribution utilizing the quantum circuits discussed above. The initialization of the class sets up a quantum device, defines the number of qubits, the number of layers, and the quantum circuit. An option to use rotations is also provided, which can be turned on or off with the \texttt{useRot} flag. If the \texttt{useRot} flag is set, the generator employs the \texttt{QuantumRandomRotationGenerator} to produce rotation angles that are subsequently fed into the  second quantum circuit. This circuit applies Hadamard gates to each qubit, creating a superposition of states, and then applies the rotations (if the \texttt{useRot} flag is set). Each qubit is then measured in the Pauli-Z basis. The measurements produce a binary string, which is converted to an index using the \texttt{toBinaryIndex} method. This method iterates over the binary outcomes, treating each binary value as a digit in a binary number and summing the corresponding powers of 2 to produce a decimal number. The \texttt{Q\_uniform\_rnd} method runs the process, generating a random number according to a uniform distribution. If the \texttt{isIndex} flag is set, it returns the binary index directly. Otherwise, it normalizes the index to a value between 0 and 1 by dividing by the total number of possible states (the space span), effectively producing a continuous uniform random number. This process exemplifies the potential of quantum computing in generating truly random numbers and highlights the flexibility of quantum circuits in constructing complex probability distributions.

\section{Results}
To test our quantum generative modeling techniques, we ran experiments both on a simulator using classical computing resources and on actual quantum hardware from IBM with 20 qubits. Running on the simulator allows us to validate the theoretical performance of the quantum circuits before real-world execution. We then benchmarked the same experiments on IBM's quantum computer to compare the practical results and quantify the error rates. While the simulator produces idealized noise-free outputs, the real hardware introduces errors from effects like decoherence and gate imperfections. By conducting experiments under both simulated and real conditions, we can verify the validity of our quantum approach and characterize the degree of deviation on current noisy devices. Analyzing the discrepancies between these results is crucial for determining the readiness of near-term quantum computers for practical generative modeling. Our evaluation provides an empirical demonstration of running quantum generation algorithms and contrasts the accuracy achieved on simulators versus real quantum chips.
\subsection{Statistical Evaluation of Quantum-Generated Distributions}
Our investigation confirms that quantum circuits are capable of generating both continuous uniform and Gaussian distributions. When employed as an alternative to classical random number generators, quantum noise effectively corrupts images for diffusion processes. The integrity of the quantum-generated Gaussian distribution is corroborated through a series of statistical tests \cite{koliada2014statistical, lorek2019testing,vattulainen1993hidden}, including a statistical permutation test, as summarized in Table (\ref{tbl:ks_test}).

\begin{table}[h!]
	\centering
	\begin{tabular}{|c|c|}
		\hline
		\textbf{Test} & \textbf{Value} \\
		\hline
		KS Statistic & 0.052 \\
		KS P-Value & 0.124 \\
		MMD & 0.001 \\
		KL Divergence & 0.030 \\
		Permutation test & 0.718 \\
		\hline
	\end{tabular}
	\caption[Statistical tests]{The quantum random Gaussian generator underwent a battery of statistical tests, juxtaposing the quantum-generated samples with classical Gaussian samples. The Kolmogorov-Smirnov (KS) test yielded a statistic of \(S = 0.052\) and a \(p\)-value of \(p = 0.124\). Given that the \(p\)-value exceeds the conventional threshold of \(0.05\), the null hypothesis, which posits that the quantum and classical samples are drawn from the same continuous distribution, could not be rejected. This suggests that the quantum samples do not significantly deviate from their classical Gaussian counterparts. The Maximum Mean Discrepancy (MMD), a measure that quantifies the difference between the mean embeddings of two distributions in a Reproducing Kernel Hilbert Space (RKHS), recorded a value of \(MMD = 0.001\). This minimal MMD value indicates an inconsequential difference between the mean embeddings of the quantum and classical Gaussian distributions in the RKHS. The Kullback-Leibler (KL) divergence, which measures how much the quantum distribution diverges from the classical Gaussian distribution, was \(D_{KL} = 0.030\). This modest KL divergence value corroborates that the quantum distribution closely approximates the classical Gaussian distribution. Finally, a statistical permutation test was conducted, with high values suggesting that the quantum and classical samples are indistinguishable. These results are consistent with the findings from the KS test and MMD, further reinforcing the congruence between the quantum and classical Gaussian distributions.}
	\label{tbl:ks_test}
\end{table}
\subsection{Quantum-Augmented Gaussian Noise in Forward Diffusion Processes}
In this study, we conducted experiments with a variable number of qubits (4, 5, 6), different PennyLane devices (e.g., `qml.device('default.qubit')`), and varying numbers of shots. In PennyLane, as is the case with many quantum simulators, the number of shots dictates the number of times the circuit is executed.
\begin{figure}[h!]
	\centering
	\includegraphics[width=\linewidth]{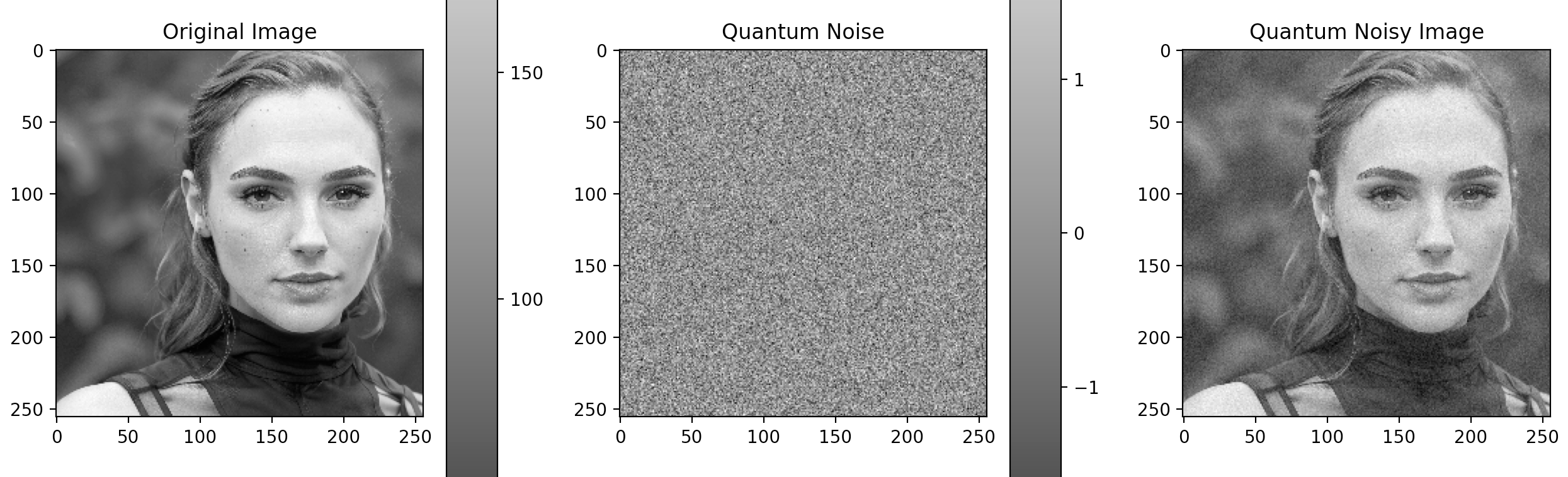}
	\caption[Quantum-Enhanced Gaussian Noise in Forward Diffusion Models]{This figure illuminates the role of Gaussian noise in forward diffusion algorithms. These algorithms convert images into pure Gaussian noise through the sequential addition of small noise perturbations. While the reverse operation, which utilizes a U-Net neural architecture to remove the noise, is not examined in this study, it plays a crucial role in reverting the image to its original, noise-free state. The efficacy of stochastic diffusion (SD) models is largely dependent on the successful integration of Gaussian noise, an aspect that could be further optimized by employing quantum random number generation (QRNG).}
	\label{fig:quantum_gaussian_2dimage}
\end{figure}
\subsection{Quantum-Enhanced Gaussian Noise in Brownian Motion Simulations}
\label{res:brownian}
This section provides an illustration of the utility of our Quantum Gaussian generator in simulating Brownian motion, serving as a quality metric for our QRNG. In conventional Brownian motion, any deviation of the Gaussian distribution's mean from zero disrupts the motion's fidelity to true Brownian behaviour (Theorem (\ref{thm:brownian})); we prove this in Proof (\ref{prf:brownian}) . Our QRNG, however, accurately encapsulates the fundamental characteristics of Brownian motion.
\begin{figure}[h!]
	\centering
	\includegraphics[width=\linewidth]{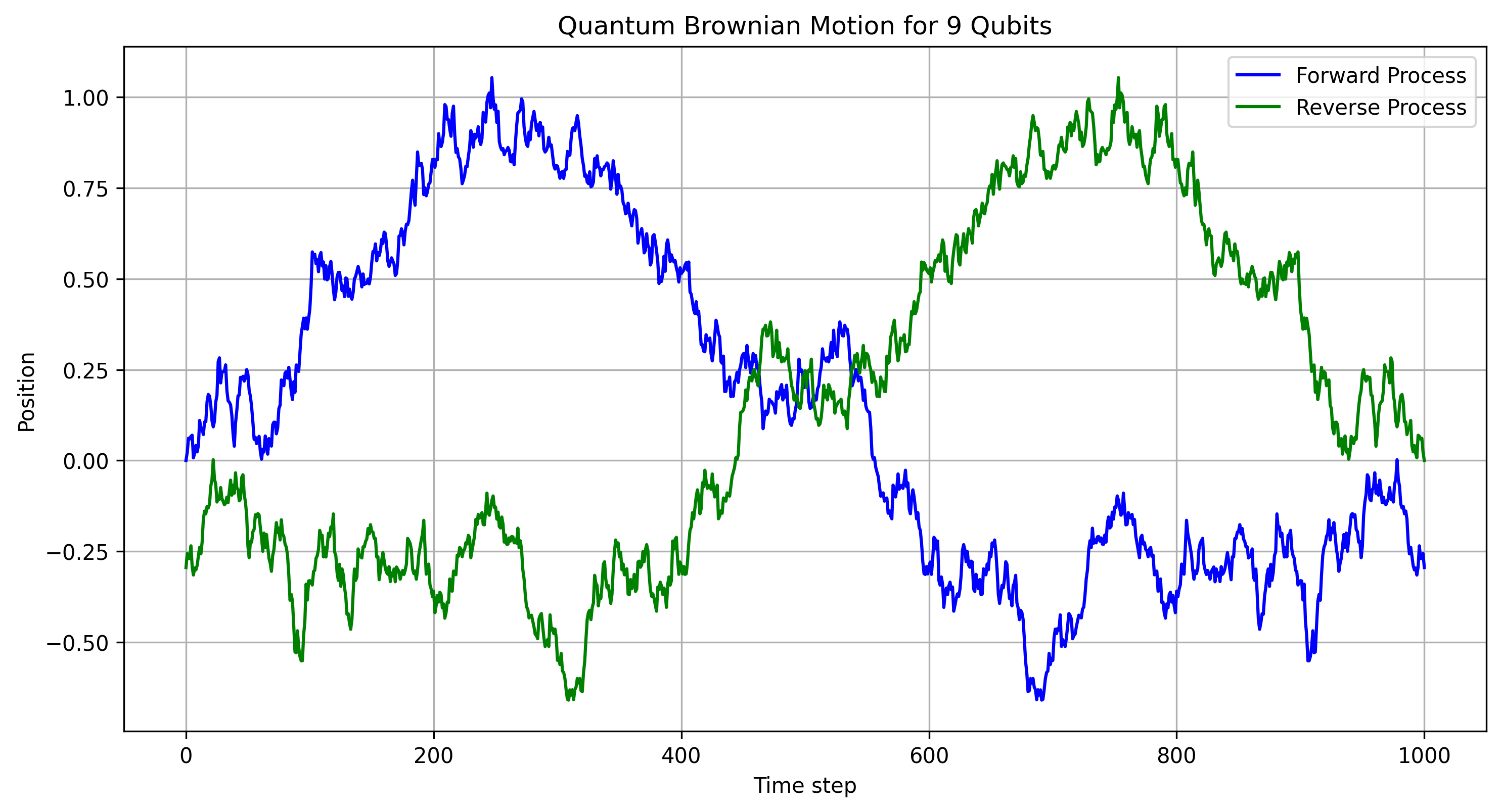}
	\caption[Quantum-Enhanced Gaussian Noise in Brownian Motion Simulations for 9 Qubits]{This figure showcases the efficacy of our Quantum Gaussian generator in simulating Brownian motion. The plot serves as a quality metric for our Quantum Random Number Generator (QRNG). Unlike conventional Brownian motion, where any deviation from a zero mean in the Gaussian distribution disrupts the fidelity to true Brownian behavior, our QRNG maintains this fidelity. The plot illustrates that the generated Brownian motion closely adheres to the theoretical expectations, thereby validating the utility of quantum-enhanced Gaussian noise in stochastic simulations.}
	\label{fig:quantum_gaussian_brownian}
\end{figure}
\begin{definition}
	Standard Brownian motion \( W(t) \) is a stochastic process that satisfies the following properties:
\begin{tabular}{ll}
	1. & \( W(0) = 0 \) \\
	2. & \( W(t) \) has independent increments. \\
	3. & \( W(t) - W(s) \sim \mathcal{N}(0, t-s) \) for \( 0 \leq s < t \). \\
	4. & \( W(t) \) is continuous in \( t \).
\end{tabular}
\end{definition}
\begin{Theorem}
	Any deviation of the Gaussian distribution's mean from zero disrupts the motion's fidelity to true Brownian behavior.
	\label{thm:brownian}
\end{Theorem}
\begin{proof}
	Let's consider a Gaussian distribution with a mean \( \mu \neq 0 \) and variance \( \sigma^2 \). If this distribution is used for the increments in Brownian motion, then the increment \( W(t) - W(s) \) would be distributed as \( \mathcal{N}(\mu, t-s) \).
	We examine each of the properties of standard Brownian motion to see if they still hold:
\begin{tabular}{ll}
	1. & \( W(0) = 0 \) still holds. \\
	2. & \( W(t) \) would still have independent increments. \\
	3. & \( W(t) - W(s) \) would now be \( \mathcal{N}(\mu, t-s) \), which violates the standard definition requiring a zero mean. \\
	4. & \( W(t) \) would still be continuous in \( t \).
\end{tabular}
	The violation occurs at the third property, thus proving that any deviation from a zero mean disrupts the fidelity to true Brownian motion.
	\label{prf:brownian}
\end{proof}
\section{Discussion}
This study presents an innovative yet straightforward strategy for generating Gaussian random variables through the utilisation of non-parametric quantum circuits. Empirical analyses substantiate a close concordance between quantum and classical Gaussian distributions, thereby affirming the efficacy of the quantum approach. Further experiments underscore the applicability of quantum-generated Gaussian noise in diffusion-centric methods such as SD and Brownian Motion.

Although the approach still necessitates classical post-processing, the replacement of pseudo-random numbers with quantum-generated random bits marks a noteworthy advancement towards more robust generative models. This hybrid methodology dispenses with the laborious task of parametric optimisation within quantum circuits. Nevertheless, the full realisation of quantum generative modelling continues to pose considerable challenges, particularly in the scaling of quantum convolutional layers. Despite these hurdles, the present study lays the groundwork for an auspicious avenue of exploration.

\section{Acknowledgements.}

\bibliographystyle{alphaurl}
\onecolumn
\listoffigures
\listoftables
\bibliography{refs}

\newcommand{\etalchar}[1]{$^{#1}$}
\begin{thebibliography}{MVDM{\etalchar{+}}22}

\bibitem[AD72]{AhrensDieter1972}
J.~H. Ahrens and U.~Dieter.
\newblock Computer methods for sampling from the exponential and normal
  distributions.
\newblock {\em Communications of the ACM}, 15(10):873--882, 1972.

\bibitem[BDT{\etalchar{+}}23]{Bai2023}
Yatong Bai, Trung Dang, Dung Tran, Kazuhito Koishida, and Somayeh Sojoudi.
\newblock Accelerating diffusion-based text-to-audio generation with
  consistency distillation, 2023.
\newblock URL: \url{https://arxiv.org/abs/2309.10740}.

\bibitem[Bil85]{billard1985computer}
Lynne Billard.
\newblock {\em Computer Science and Statistics: Proceedings of the Sixteenth
  Symposium on the Interface, Atlanta, Georgia, March 1984}.
\newblock North Holland, 1985.

\bibitem[BIS{\etalchar{+}}18]{bergholm2018pennylane}
Ville Bergholm, Josh Izaac, Maria Schuld, Christian Gogolin, Carsten Blank,
  Keri McKiernan, and Nathan Killoran.
\newblock Pennylane: Automatic differentiation of hybrid quantum-classical
  computations.
\newblock {\em arXiv preprint arXiv:1811.04968}, 2018.

\bibitem[BLSF19]{benedetti2019}
Marcello Benedetti, Erika Lloyd, Stefan Sack, and Mattia Fiorentini.
\newblock Parameterized quantum circuits as machine learning models.
\newblock {\em Quantum Science and Technology}, 4(4), Nov 2019.
\newblock \href {https://doi.org/10.1088/2058-9565/ab4eb5}
  {\path{doi:10.1088/2058-9565/ab4eb5}}.

\bibitem[BM58]{box1958note}
George~EP Box and Mervin~E Muller.
\newblock A note on the generation of random normal deviates.
\newblock {\em The annals of mathematical statistics}, 29(2):610--611, 1958.

\bibitem[Bre74]{brent1974gaussian}
Richard~P. Brent.
\newblock Gaussian pseudo-random number generator, 1974.

\bibitem[Bre10a]{brent2010fast}
Richard~P Brent.
\newblock Fast normal random number generators on vector processors.
\newblock {\em arXiv preprint arXiv:1004.3105}, 2010.
\newblock URL: \url{http://arxiv.org/abs/1004.3105v2}.

\bibitem[Bre10b]{brent2010longperiod}
Richard~P. Brent.
\newblock Some long-period random number generators using shifts and xors,
  2010.
\newblock \href {http://arxiv.org/abs/1004.3115} {\path{arXiv:1004.3115}}.

\bibitem[BWP{\etalchar{+}}17]{biamonte2017}
Jacob Biamonte, Peter Wittek, Nicola Pancotti, Patrick Rebentrost, Nathan
  Wiebe, and Seth Lloyd.
\newblock Quantum machine learning.
\newblock {\em Nature}, 549(7671), Sep 2017.
\newblock \href {https://doi.org/10.1038/nature23474}
  {\path{doi:10.1038/nature23474}}.

\bibitem[CCLZ23]{cao2023exploring}
Yu~Cao, Jingrun Chen, Yixin Luo, and Xiang Zhou.
\newblock Exploring the optimal choice for generative processes in diffusion
  models: Ordinary vs stochastic differential equations.
\newblock {\em arXiv preprint arXiv:2306.02063}, 2023.
\newblock URL: \url{http://arxiv.org/abs/2306.02063v1}.

\bibitem[CDL18]{cheng2018quantum}
Xi~Cheng, Yudong Du, and Bo~Liu.
\newblock Quantum generative models for small molecule drug discovery.
\newblock {\em arXiv preprint arXiv:1806.01423}, 2018.

\bibitem[Cha43]{Chandrasekhar1943}
Subrahmanyan Chandrasekhar.
\newblock Stochastic problems in physics and astronomy.
\newblock {\em Reviews of Modern Physics}, 15:1–89, 1943.

\bibitem[DN21]{dhariwal2021diffusion}
Prafulla Dhariwal and Alex Nichol.
\newblock Diffusion models beat gans on image synthesis.
\newblock {\em arXiv preprint arXiv:2105.05233}, 2021.
\newblock URL: \url{http://arxiv.org/abs/2105.05233v4}.

\bibitem[Ein05]{Einstein1905}
Albert Einstein.
\newblock On the movement of small particles suspended in stationary liquids
  required by the molecular-kinetic theory of heat.
\newblock {\em Annalen der Physik}, 17:549–560, 1905.

\bibitem[FGG14]{farhi2014quantum}
Edward Farhi, Jeffrey Goldstone, and Sam Gutmann.
\newblock A quantum approximate optimization algorithm, 2014.

\bibitem[FKB{\etalchar{+}}23]{Fishman2023}
Nic Fishman, Leo Klarner, Valentin~De Bortoli, Emile Mathieu, and Michael
  Hutchinson.
\newblock Diffusion models for constrained domains, 2023.
\newblock URL: \url{https://arxiv.org/abs/2304.05364}.

\bibitem[FRY{\etalchar{+}}22]{franzese2022how}
Giulio Franzese, Simone Rossi, Lixuan Yang, Alessandro Finamore, Dario Rossi,
  Maurizio Filippone, and Pietro Michiardi.
\newblock How much is enough? a study on diffusion times in score-based
  generative models.
\newblock {\em arXiv preprint arXiv:2206.05173}, 2022.
\newblock URL: \url{http://arxiv.org/abs/2206.05173v1}.

\bibitem[Gar04]{Gardiner2004}
Crispin~W. Gardiner.
\newblock {\em Handbook of Stochastic Methods}.
\newblock Springer, 2004.

\bibitem[GKSB23]{gili2023generative}
Kaitlin Gili, Rohan~S. Kumar, Mykolas Sveistrys, and C.~J. Ballance.
\newblock Generative modeling with quantum neurons.
\newblock {\em arXiv preprint arXiv:2302.00788}, 2023.
\newblock URL: \url{http://arxiv.org/abs/2302.00788v1}.

\bibitem[GMPO22]{gili2022evaluating}
Kaitlin Gili, Marta Mauri, and Alejandro Perdomo-Ortiz.
\newblock Evaluating generalization in classical and quantum generative models.
\newblock {\em arXiv preprint arXiv:2201.08770}, 2022.

\bibitem[GSLW19]{qsvt}
András Gilyén, Yuan Su, Guang~Hao Low, and Nathan Wiebe.
\newblock Quantum singular value transformation and beyond: exponential
  improvements for quantum matrix arithmetics.
\newblock In {\em ACM SIGACT Symposium on Theory of Computing}, 2019.
\newblock \href {https://doi.org/10.1145/3313276.3316366}
  {\path{doi:10.1145/3313276.3316366}}.

\bibitem[HJA20]{Ho2020}
Jonathan Ho, Ajay Jain, and Pieter Abbeel.
\newblock Denoising diffusion probabilistic models.
\newblock {\em Advances in Neural Information Processing Systems},
  33:6840--6851, 2020.

\bibitem[HZ13]{Hwang2013}
Hsien-Kuei Hwang and Vytas Zacharovas.
\newblock Limit laws of the coefficients of polynomials with only unit roots,
  2013.
\newblock URL: \url{http://arxiv.org/abs/1301.2021v1}.

\bibitem[KMT{\etalchar{+}}17]{kandala2017hardware}
Abhinav Kandala, Antonio Mezzacapo, Kristan Temme, Maika Takita, Markus Brink,
  Jerry~M Chow, and Jay~M Gambetta.
\newblock Hardware-efficient variational quantum eigensolver for small
  molecules and quantum magnets.
\newblock {\em Nature}, 549(7671):242--246, 2017.

\bibitem[Kol14]{koliada2014statistical}
Sergii Koliada.
\newblock Statistical tests for a sequence of random numbers by using the
  distribution function of random distance in three dimensions.
\newblock {\em arXiv preprint arXiv:1402.5383}, 2014.
\newblock URL: \url{http://arxiv.org/abs/1402.5383v1}.

\bibitem[Lan08]{Langevin1908}
Paul Langevin.
\newblock On the theory of brownian motion.
\newblock {\em Comptes Rendus}, 146:530–533, 1908.

\bibitem[LLGZ19]{lorek2019testing}
Pawel Lorek, Grzegorz Los, Karol Gotfryd, and Filip Zagorski.
\newblock On testing pseudorandom generators via statistical tests based on the
  arcsine law.
\newblock {\em arXiv preprint arXiv:1903.09805}, 2019.
\newblock URL: \url{http://arxiv.org/abs/1903.09805v1}.

\bibitem[Luo22]{Luo2022}
Cathy Luo.
\newblock Understanding diffusion models: A unified perspective.
\newblock {\em arXiv preprint arXiv:2206.00277}, 2022.

\bibitem[LW18]{liu2018differentiable}
Jin-Guo Liu and Lei Wang.
\newblock Differentiable learning of quantum circuit born machine.
\newblock {\em arXiv preprint arXiv:1804.04168}, 2018.
\newblock URL: \url{http://arxiv.org/abs/1804.04168v1}.

\bibitem[MB22]{Baas2022}
Herman~Kamper Matthew~Baas, Kevin~Eloff.
\newblock Transfusion: Transcribing speech with multinomial diffusion, 2022.
\newblock URL: \url{https://arxiv.org/abs/2210.07677}.

\bibitem[MRBAG16]{mcclean2016theory}
Jarrod~R McClean, Jonathan Romero, Ryan Babbush, and Al{\'a}n Aspuru-Guzik.
\newblock The theory of variational hybrid quantum-classical algorithms, 2016.

\bibitem[MS18]{Michelen2018}
Marcus Michelen and Julian Sahasrabudhe.
\newblock Central limit theorems from the roots of probability generating
  functions, 2018.
\newblock URL: \url{http://arxiv.org/abs/1804.07696v2}.

\bibitem[MVDM{\etalchar{+}}22]{moghadam2022morphology}
Puria~Azadi Moghadam, Sanne Van~Dalen, Karina~C. Martin, Jochen Lennerz,
  Stephen Yip, Hossein Farahani, and Ali Bashashati.
\newblock A morphology focused diffusion probabilistic model for synthesis of
  histopathology images.
\newblock {\em arXiv preprint arXiv:2209.13167}, 2022.
\newblock URL: \url{http://arxiv.org/abs/2209.13167v2}.

\bibitem[Nad06]{Nadler2006}
Boaz Nadler.
\newblock Design flaws in the implementation of the ziggurat and monty python
  methods (and some remarks on matlab randn), 2006.
\newblock URL: \url{http://arxiv.org/abs/math/0603058v1}.

\bibitem[NC10]{nielsen_chuang_2010}
Michael~A. Nielsen and Isaac~L. Chuang.
\newblock {\em Quantum Computation and Quantum Information: 10th Anniversary
  Edition}.
\newblock Cambridge University Press, 2010.

\bibitem[Per08]{Perrin1908}
Jean Perrin.
\newblock Brownian movement and molecular reality.
\newblock {\em Annals of Physics}, 18:187–214, 1908.

\bibitem[PMS{\etalchar{+}}14]{peruzzo2014variational}
Alberto Peruzzo, Jarrod McClean, Peter Shadbolt, Man-Hong Yung, Xiao-Qi Zhou,
  Peter~J Love, Al{\'a}n Aspuru-Guzik, and Jeremy~L O'Brien.
\newblock A variational eigenvalue solver on a photonic quantum processor.
\newblock {\em Nature communications}, 5:4213, 2014.

\bibitem[RAG19]{romero2019variational}
Jonathan Romero and Alan Aspuru-Guzik.
\newblock Variational quantum generators: Generative adversarial quantum
  machine learning for continuous distributions.
\newblock {\em arXiv preprint arXiv:1901.00848}, 2019.
\newblock URL: \url{http://arxiv.org/abs/1901.00848v1}.

\bibitem[RFB15]{Ronneberger2015}
Olaf Ronneberger, Philipp Fischer, and Thomas Brox.
\newblock U-net: Convolutional networks for biomedical image segmentation.
\newblock In {\em International Conference on Medical image computing and
  computer-assisted intervention}, pages 234--241. Springer, 2015.

\bibitem[SAH22]{AbuHussein2022}
Raja~Giryes Shady Abu-Hussein, Tom~Tirer.
\newblock Adir: Adaptive diffusion for image reconstruction, 2022.
\newblock URL: \url{https://arxiv.org/abs/2212.03221}.

\bibitem[SDME21]{Song2021}
Yang Song, Conor Durkan, Ian Murray, and Stefano Ermon.
\newblock Maximum likelihood training of score-based diffusion models.
\newblock In {\em Advances in Neural Information Processing Systems},
  volume~34, 2021.

\bibitem[SDWMG15]{Sohl-Dickstein2015}
Jascha Sohl-Dickstein, Eric~A Weiss, Niru Maheswaranathan, and Surya Ganguli.
\newblock Deep unsupervised learning using nonequilibrium thermodynamics.
\newblock {\em arXiv preprint arXiv:1503.03585}, 2015.

\bibitem[SE19]{Song2019}
Yang Song and Stefano Ermon.
\newblock Generative modeling by estimating gradients of the data distribution.
\newblock In {\em Advances in Neural Information Processing Systems}, pages
  11871--11881, 2019.

\bibitem[SSDK{\etalchar{+}}20a]{song2020score}
Yang Song, Jascha Sohl-Dickstein, Diederik~P. Kingma, Abhishek Kumar, Stefano
  Ermon, and Ben Poole.
\newblock Score-based generative modeling through stochastic differential
  equations.
\newblock {\em arXiv preprint arXiv:2011.13456}, 2020.
\newblock URL: \url{http://arxiv.org/abs/2011.13456v2}.

\bibitem[SSDK{\etalchar{+}}20b]{Song2020}
Yang Song, Jascha Sohl-Dickstein, Diederik~P Kingma, Abhishek Kumar, Stefano
  Ermon, and Ben Poole.
\newblock Score-based generative modeling through stochastic differential
  equations.
\newblock {\em arXiv preprint arXiv:2011.13456}, 2020.

\bibitem[Vig14]{Vigna2014b}
Sebastiano Vigna.
\newblock Further scramblings of marsaglia's xorshift generators, 2014.
\newblock URL: \url{http://arxiv.org/abs/1404.0390v3}.

\bibitem[vK07]{vanKampen2007}
N.G. van Kampen.
\newblock {\em Stochastic Processes in Physics and Chemistry}.
\newblock Elsevier, 2007.

\bibitem[VKSAN93]{vattulainen1993hidden}
I~Vattulainen, K~Kankaala, J~Saarinen, and T~Ala-Nissila.
\newblock Hidden errors in simulations and the quality of pseudorandom numbers.
\newblock {\em arXiv preprint cond-mat/9301022}, 1993.
\newblock URL: \url{http://arxiv.org/abs/cond-mat/9301022v1}.

\bibitem[YS07]{Yahav2007}
Inbal Yahav and Galit Shmueli.
\newblock An elegant method for generating multivariate poisson random
  variable, 2007.
\newblock URL: \url{http://arxiv.org/abs/0710.5670v2}.

\bibitem[YYT{\etalchar{+}}22]{Yuan2022}
Hongyi Yuan, Zheng Yuan, Chuanqi Tan, Fei Huang, and Songfang Huang.
\newblock Seqdiffuseq: Text diffusion with encoder-decoder transformers, 2022.
\newblock URL: \url{https://arxiv.org/abs/2212.10325}.

\bibitem[YZS{\etalchar{+}}22]{Yang2022}
Ling Yang, Zhichao Zhang, Yang Song, Shiwei Hong, Runze Xu, Yan Zhao, Wei
  Zhang, Bingzhen Cui, and Ming-Hsuan Yang.
\newblock Diffusion models: A comprehensive survey of methods and applications.
\newblock {\em arXiv preprint arXiv:2206.08773}, 2022.

\end{thebibliography}

\renewcommand\indexname{Index}
\end{document}